%% file: main.tex
\documentclass[11pt]{article}

\input{header/arxiv_header}

\input{header/custom_header}

\begin{document}
\input{title/title}

\maketitle

\input{title/abstract}

\input{paper/intro}
\input{paper/preliminary}
\input{algorithms/algorithm-l1}
\input{paper/proof-overview}
\input{paper/L1-k-medians}
\input{paper/L1-improved}
\input{paper/terminal-embed-L2L1}
\input{paper/L2-k-medians}
\input{paper/lower-bound}

\input{paper/fast-algorithm}
\bibliographystyle{plainnat}
\bibliography{ref}
\end{document}

%% file: header/arxiv_header.tex
\usepackage{fullpage}
\usepackage[round]{natbib}
\usepackage{algorithmic}
\usepackage{framed}
\usepackage{microtype}
\usepackage{graphicx}
\usepackage{subfigure}
\usepackage{booktabs}
\usepackage[hyperindex,breaklinks]{hyperref}

\usepackage{amsmath,amsthm,amssymb,amsfonts}

\usepackage{mathtools}
\usepackage{dsfont}

\usepackage{tikz,pgfplots}
\pgfplotsset{compat=1.16}

\usepackage{nicefrac}

\usepackage{silence}


%% file: header/custom_header.tex
\newcommand{\E}{\mathbb{E}}

\newcommand{\pr}{\mathbb{P}}
\newcommand{\prob}[1]{\pr\left[#1\right]}
\newcommand{\abs}[1]{\left\lvert#1\right\rvert}

\newcommand{\Phase}{\mathrm{Phase}}

\def\rbr#1{\left(#1\right)}

\def\cbr#1{\left\{#1\right\}}
\DeclarePairedDelimiter\ceil{\lceil}{\rceil}
\DeclarePairedDelimiter\floor{\lfloor}{\rfloor}

\def\norm#1{\lVert#1\rVert}
\DeclarePairedDelimiterX{\inner}[2]{\langle}{\rangle}{#1, #2}

\newcommand{\real}{\mathbb{R}}
\newcommand{\ind}{\mathds{1}}

\newtheorem{theorem}{Theorem}[section]
\newtheorem{claim}[theorem]{Claim}

\newtheorem{lemma}[theorem]{Lemma}

\DeclareMathOperator{\sign}{sign}

\newcommand{\alg}{\mathrm{alg}}
\newcommand{\cost}{\mathrm{cost}}
\newcommand{\OPT}{\mathrm{OPT}}

\newcommand{\bbR}{\mathbb{R}}

\newcommand{\calF}{\mathcal{F}}

\newcommand{\calP}{\mathcal{P}}

%% file: title/title.tex
\title{\texorpdfstring{Near-optimal Algorithms for Explainable $k$-Medians and $k$-Means}{Near-optimal Algorithms for Explainable k-Medians and k-Means}\footnote{The conference version of this paper appeared in the proceedings of ICML 2021.}}
\author{
Konstantin Makarychev\footnote{Equal  contribution. The authors were supported by
NSF Awards CCF-1955351 and CCF-1934931.}
\and 
Liren Shan\footnotemark[2]}
\date{Northwsestern University}

%% file: title/abstract.tex
\begin{abstract}
We consider the problem of explainable $k$-medians and $k$-means introduced by Dasgupta, Frost, Moshkovitz, and Rashtchian~(ICML 2020). 
In this problem, our goal is to find a \emph{threshold decision tree} that partitions data into $k$ clusters and minimizes the $k$-medians or $k$-means objective. The obtained clustering is easy to interpret because every decision  node of a threshold tree splits data  based on a single feature into two groups. We propose a new algorithm for this problem which is $\tilde O(\log k)$ competitive with $k$-medians with $\ell_1$ norm and $\tilde O(k)$ competitive with $k$-means. This is an improvement over the previous guarantees of $O(k)$ and $O(k^2)$ by Dasgupta et al (2020). We also provide a new algorithm which is $O(\log^{\nicefrac{3}{2}} k)$  competitive for $k$-medians with $\ell_2$ norm. 
Our first algorithm is near-optimal: Dasgupta et al (2020) showed a lower bound of $\Omega(\log k)$ for $k$-medians; in this work, we prove a lower bound of $\tilde\Omega(k)$ for 
$k$-means. We also provide a lower bound of $\Omega(\log k)$ for $k$-medians with $\ell_2$ norm.
\end{abstract}

%% file: paper/intro.tex
\section{Introduction}
In this paper, we investigate the problem of \emph{explainable} 
$k$-means and $k$-medians clustering which was recently introduced
by~\citet*{DFMR20}. Suppose, we have a data set which we need to partition into $k$ clusters. How can we do it? 
Of course, we could use one of many standard algorithms for $k$-means or $k$-medians clustering. However, we want to find an \emph{explainable} clustering
-- clustering which can be easily understood by a human being. Then, $k$-means or $k$-medians clustering may not be the best options for us.

Note that though every cluster in a $k$-means and $k$-medians clustering has a simple mathematical description, this description is 
not necessarily easy to interpret for a human. Every $k$-medians or $k$-means  clustering is defined by a set of $k$ centers $c^1,c^2,\dots, c^k$, where
each cluster is the set of points located closer to a fixed center $c^i$ than to any other center $c^j$. That is, for points in cluster $i$, we must 
have $\arg\min_{j}\|x-c^j\| = i$. Thus, in order to determine
to which cluster a particular point belongs, we need to compute
distances from point $x$ to all centers $c^j$. Each distance depends
on all coordinates of the points. Hence, for a human, it is not 
even easy to figure out to which cluster in $k$-means or $k$-medians  clustering a particular point belongs to; let alone interpret the entire  clustering.

In every day life, we are surrounded by different types of classifications. Consider the following examples from Wikipedia: 
\emph{(1) Performance cars are capable of going from 0 to 60 mph in under 5 seconds}; \emph{(2) Modern sources currently define skyscrapers as being at least 100 metres or 150 metres in height;} \emph{(3) Very-low-calorie diets are diets of 800 kcal or less energy intake per day, whereas low-calorie diets are between 1000-1200 kcal per day}. Note that all these 
definitions depend on a \emph{single feature} which makes them easy 
to understand. 

The above discussion leads us to the idea of \citet{DFMR20}, who proposed to use threshold (decision) trees to describe clusters
(see also \citet*{liu2005clustering}, \citet*{fraiman2013interpretable}, \citet*{bertsimas2018interpretable}, and \citet*{saisubramanian2020balancing}). 

A threshold tree is a binary classification tree with $k$ leaves. Every internal node $u$ of the tree splits the data into two sets by comparing a single feature $i_u$ of each data point with a threshold $\theta_u$. The first set is the set of points with $x_{i_u}\leq \theta_u$; the second set is the set of points
with $x_{i_u} > \theta_u$. These two sets are then recursively partitioned 
by the left and right children of $u$. Thus, each point $x$ in the data set is eventually assigned to one of $k$ leaves of the threshold tree $T$. This gives us a partitioning of the data set $X$ into clusters 
$\calP = (P_1,\dots, P_k)$. We note that threshold decision trees are special cases of binary space partitioning  (BSP) trees  and similar to $k$-d trees~\cite{kdTrees}.

\citet{DFMR20} suggested that we measure the quality of a threshold tree using the standard $k$-means and $k$-medians objectives. 
Specifically, the $k$-medians in $\ell_1$ cost of the threshold tree $T$ equals (\ref{objective:L1}),
the $k$-medians in $\ell_2$ cost equals (\ref{objective:L2})
and $k$-means cost  equals (\ref{objective:L22}):
\begin{align}
\label{objective:L1}
\cost_{\ell_1}(X,T) &= \sum_{i=1}^k \sum_{x \in P_i} \norm{x-c^i}_1,\\
\label{objective:L2}
\cost_{\ell_2}(X,T) &= \sum_{i=1}^k \sum_{x \in P_i} \norm{x-c^i}_2,\\
\label{objective:L22}
\cost_{\ell_2^2}(X,T) &= \sum_{i=1}^k \sum_{x \in P_i} \norm{x-c^i}_2^2,
\end{align}
where $c^i$ is the $\ell_1$-median of cluster $P_i$ in (\ref{objective:L1}), the $\ell_2$-median of cluster $P_i$ in (\ref{objective:L2}), and the mean of cluster $P_i$ in (\ref{objective:L22}).

This definition raises obvious questions: Can we actually find a good explainable clustering? Moreover, how good can it be comparing to a regular $k$-medians and $k$-means clustering? Let $\OPT_{\ell_1}(X)$, $\OPT_{\ell_2}(X)$, and $\OPT_{\ell_2^2}(X)$ be the optimal solutions to (regular) 
$k$-medians with $\ell_1$ norm, $k$-medians with $\ell_2$ norm, and $k$-means, respectively. 
\citet{DFMR20} defined 
the \emph{price of explainability} as the ratio
$\cost_{\ell_1}(X,T)/\OPT_{\ell_1}(X)$ for $k$-medians in $\ell_1$ and $\cost_{\ell_2^2}(X,T)/\OPT_{\ell_2^2}(X)$ for
$k$-means. The price of explainability shows by how much the optimal 
unconstrained solution is better than the best explainable solution 
for the same data set. 

In their paper, \citet{DFMR20} gave upper and lower bounds on the price of explainability. They proved that the price of explainability
is upper bounded by $O(k)$ and $O(k^2)$ for $k$-medians in $\ell_1$ and $k$-means, respectively. Furthermore, they designed two algorithms that given a
$k$-medians in $\ell_1$ or $k$-means clustering, produce an explainable clustering 
with cost at most $O(k)$ and $O(k^2)$ times the cost of original clustering (respectively). They also provided examples for which the price of explainability of $k$-medians in $\ell_1$ and $k$-means is at least $\Theta(\log k)$.

\subsection{Our results}
In this work, we give almost tight bounds on the price of 
explainability for both $k$-medians in $\ell_1$ and $k$-means. Specifically,
we show how to transform any clustering to an explainable clustering with cost at most $O(\log k \log\log k)$ times the original cost for the $k$-medians $\ell_1$ objective and $O(k \log k \log\log k)$ for the $k$-means objective.
Note that we get an exponential improvement over previous results for the $k$-medians $\ell_1$
objective. Furthermore, we present an algorithm for $k$-medians in $\ell_2$ with the price of explainability bounded by $O(\log^{\nicefrac{3}{2}} k)$. We complement these results with an almost tight lower
bound of $\Omega(k/\log k)$ for the $k$-means objective and an $\Omega(\log k)$ lower bound for $k$-medians in $\ell_2$ objective. 
We summarise our results in Table~\ref{table-our-results}.

\input{figures/table-of-results}

Below, we formally state our main results.
The costs of threshold trees and clusterings are defined by formulas (\ref{objective:L1}), (\ref{objective:L2}), (\ref{objective:L22}), 
(\ref{def:cost1}), (\ref{def:cost2}), and (\ref{def:cost2^2}).
\begin{theorem}\label{thm:kmedians-improved}
There exists a polynomial-time randomized algorithm that given a data set $X$ and a set of centers $C = \{c^1,\dots,c^k\}$, finds a threshold tree $T$ with expected $k$-medians in $\ell_1$ cost at most
$$
\E[\cost_{\ell_1}(X,T)] \leq O(\log k\log\log k)\cdot \cost_{\ell_1}(X,C).
$$
\end{theorem}
\begin{theorem}\label{thm:k-means}
There exists a polynomial-time randomized algorithm that given a data set $X$ and a set of centers $C = \{c^1,\dots,c^k\}$, finds a threshold
tree $T$ with expected $k$-means cost at most
$$
\E[\cost_{\ell_2^2}(X,T)] \leq O(k \log k \log\log k)\cdot \cost_{\ell_2^2}(X,C).
$$
\end{theorem}
We note that the algorithms by \citet{DFMR20} also produce trees based on the given set of ``reference'' centers $c^1,\dots,c^k$. However, the approximation guarantees of those algorithms are $O(k)$ and $O(k^2)$, respectively.
Our upper bound of $O(\log k\log \log k)$
almost matches the lower bound of $\Omega(\log k)$
given by~\citet{DFMR20}. The upper bound of $O(k\log k \log\log k)$
almost matches the lower bound of $\Omega(k/\log k)$ we show in Section~\ref{sec:lb-kmeans}.

\begin{theorem}\label{thm:k-medians-L2}
There exists a polynomial-time randomized algorithm that given a data set $X$ and a set of centers $C = \{c^1,\dots,c^k\}$, finds a threshold tree $T$ with expected $k$-medians in $\ell_2$ cost at most
$$
\E[\cost_{\ell_2}(X,T)] \leq O(\log^{\nicefrac{3}{2}} k)\cdot \cost_{\ell_2}(X,C).
$$
\end{theorem}

\subsection{Related work}
\citet{DFMR20} introduced the \emph{explainable} $k$-medians and $k$-means clustering problems and developed Iterative Mistake Minimization (IMM) algorithms
for these problems. Later, \citet*{frost2020exkmc} proposed algorithms that construct threshold trees with more than $k$ leaves.

Decision trees have been used for interpretable classification and clustering since 1980s. \citet*{breiman1984classification} proposed a popular decision tree 
algorithm called CART for supervised classification. For unsupervised clustering, threshold decision trees are used in many empirical methods based on different criteria such as information gain \cite{liu2005clustering}, local $1$-means cost \cite{fraiman2013interpretable},
Silhouette Metric \cite{bertsimas2018interpretable}, and interpretability score \cite{saisubramanian2020balancing}. 

The $k$-means and $k$-medians clustering problems have been extensively studied in the literature. The $k$-means++ algorithm proposed by \citet*{arthur2006k} is the most widely used algorithm for $k$-means clustering. It provides an $O(\ln k)$ approximation. \citet*{LiSvensson16} provided a $1+\sqrt{3} + \varepsilon$ approximation for $k$-medians in general metric spaces,
which was improved to $2.611 +\varepsilon$ by \citet*{bprst14}. \citet*{ahmadian2019better} gave a $6.357$ approximation algorithm for $k$-means. The $k$-medians and $k$-means problems are NP-hard \cite{MS1984,dasgupta2008hardness,aloise2009np}. 
Recently, \citet*{awasthi2015hardness} showed that it is also NP-hard to approximate the $k$-means objective within a factor of $(1+\varepsilon)$ for some positive constant $\varepsilon$ (see also \citet{LSW17}). \citet*{bhattacharya2020hardness} showed that the Euclidean $k$-medians can not be approximated within a factor of $(1+\varepsilon)$ for some constant $\varepsilon$ assuming the unique games conjecture.

\citet{bmd09}, \citet{BZMD14}, \citet{CEMMP15}, \citet{MMR19} and \citet{BBCGS19}
showed how to reduce the dimensionality of a data set for $k$-means clustering. Particularly, \citet{MMR19} proved that we can use the Johnson--Lindenstrauss transform to reduce the dimensionality of $k$-medians in $\ell_2$ and $k$-means to $d'= O(\log k)$. Note, however, that the Johnson--Lindenstrauss transform cannot be used for the explainable $k$-medians and $k$-means problems, because this transform does not preserve the set of features. Instead, we can use a \emph{feature selection} algorithm by~\citet{BZMD14} or \citet{CEMMP15} to reduce the dimensionality to 
$d' = \tilde{O}(k)$.

Independently of our work, \citet*{laber2021price} proposed new algorithms for explainable $k$-medians with $\ell_1$ and $k$-means objectives. Their competitive ratios are $O(d\log k)$ and $O(dk\log k)$, respectively. Note that these competitive ratios depend on the dimension $d$ of the space.

\bigskip

\noindent\textbf{Remark: } After this paper was accepted to ICML 2021, we learned about three independent results that were recently posted on arXiv.org. The first paper by~\citet{CharikarHu21} gives a
$k^{1-2/d}\operatorname{poly}(d\log k)$-competitive algorithm for $k$-means. Note that this bound depends on the dimension of the data set. It is better than our $k$-means bound ($O(k \log k \log\log k)$) for small $d$ ($d \ll \log k /\log\log k$) and worse for large $d$ ($d \gg \log k /\log\log k$). The second paper by~\citet*{GJPS21} gives $O(\log^2 k)$ and $O(k \log^2 k)$-competitive algorithms for $k$-medians in $\ell_1$ and $k$-means, respectively. These bounds are slightly worse than ours. The third paper by~\citet*{EMN21} gives the same algorithm as ours for $k$-medians in $\ell_1$. They show that this algorithm is both $O(\log k \log\log k)$-competitive and $O(d\log^2 d)$-competitive. Note that the second bound is better than our bound in low-dimensional spaces, when $d < \log k$. They also provide an $O(k \log k)$-competitive algorithm and an $\Omega(k)$ lower bound for $k$-means, which are slightly better than ours. 

%% file: figures/table-of-results.tex
\begin{figure*}
\renewcommand{\arraystretch}{1.2}
\begin{center}
\scalebox{0.85}{
\begin{tabular}{|c |c | c | c | c |c | c | c |}
\hline
 & \multicolumn{2}{|c|}{\textbf{$k$-medians in $\ell_1$}}
 & \multicolumn{2}{|c|}{\textbf{$k$-medians in $\ell_2$}} 
 & \multicolumn{2}{|c|}{\textbf{$k$-means}}\\
\hline\hline   
 & \textbf{Lower} & \textbf{Upper} 
 & \textbf{Lower} & \textbf{Upper}
 & \textbf{Lower} & \textbf{Upper}\\ 
\hline\hline   
\textbf{Our results}    
& & $O(\log k \log\log k)$ & $\Omega(\log k)$ & $O(\log^{\nicefrac{3}{2}} k)$ & $\Omega(k/\log k)$ & $O(k\log k\log\log k)$\\ 
\hline
\textbf{\citet{DFMR20}} 
& $\Omega(\log k)$ & $O(k)$ &  &  & $\Omega(\log k)$ & $O(k^2)$ \\
\hline
\end{tabular}}
\caption{Summary of our results. The table shows known upper and lower bounds on the \textit{price of explainability} for $k$-medians in $\ell_1$ and $\ell_2$, and for $k$-means.}
\label{table-our-results}
\end{center}
\end{figure*}

%% file: paper/preliminary.tex
\section{Preliminaries}

Given a set of points $X \subseteq{\bbR^d}$ and an integer $k>1$, the regular $k$-medians and $k$-means clustering problems are to find a set $C$ of $k$ centers to minimize the corresponding costs: $k$-medians with $\ell_1$ objective cost~(\ref{def:cost1}), $k$-medians with $\ell_2$ objective cost~(\ref{def:cost2}), and $k$-means cost~(\ref{def:cost2^2}).
\begin{align}
\label{def:cost1}
\cost_{\ell_1}(X,C) = \sum_{x\in X} \min_{c \in C}\norm{x_i-c}_1,
\\
\label{def:cost2}
\cost_{\ell_2}(X,C) = \sum_{x\in X} \min_{c \in C}\norm{x_i-c}_2.
\end{align}
\begin{align}
\label{def:cost2^2}
\mathrm{cost}_{\ell_2^2}(X,C) = \sum_{x\in X} \min_{c \in C}\norm{x_i-c}_2^2.
\end{align}

Every coordinate cut is specified by the coordinate $i\in\{1,\dots,d\}$ and threshold $\theta$. We denote the set of all possible cuts by $\Omega$:
$$\Omega = \{1,\cdots, d\} \times \bbR.$$
We define the standard product measure on $\Omega$ as follows:
The measure of set $S\subset \Omega$ equals
$$\mu(S)=\sum_{i=1}^d \mu_{R}(\{\theta: (i,\theta)\in S\}),$$
where $\mu_R$ is the Lebesgue measure on $\bbR$.

For every cut $\omega = (i,\theta)\in \Omega$ and point $x\in \bbR^d$, we let 
$$\delta_x(\omega) \equiv  \delta_x(i, \theta) = 
\begin{cases}
1,&\text{if } x_i > \theta;\\
0,&\text{otherwise.}
\end{cases}
$$
In other words, $\delta_x(i, \theta)$ is the indicator of the event $\{x_i > \theta\}$.
Observe that $x\mapsto \delta_x$ is an isometric embedding of $\ell^d_1$ ($d$-dimensional $\ell_1$ space)
into $L_1(\Omega)$ (the space of  integrable functions on $\Omega$). Specifically, for $x,y\in \bbR^d$, we have
\begin{align}\label{eq:isometry}
\notag
\|x-y\|_1 &\equiv \sum_{i=1}^d |x_i - y_i| \\
&= \sum_{i=1}^d \int\displaylimits_{-\infty}^{\infty} |\delta_x(i,\theta) - \delta_y(i,\theta)|\; d\theta\\ 
\notag
&=\int_{\Omega} |\delta_x(\omega) - \delta_y(\omega)| \; d\mu(\omega)\equiv \|\delta_x - \delta_y\|_1.
\end{align}

A map $\varphi:\bbR^d \to \bbR^d$ is coordinate cut preserving if for every coordinate cut $(i,\theta) \in \Omega$, there exists a coordinate cut $(i',\theta')\in \Omega$ such that $\{x\in \bbR^d: x_{i'} \leq \theta'\} = \{x\in \bbR^d: \varphi(x)_i \leq \theta\}$ and vice versa. In the algorithm for explainable $k$-means, we use a cut preserving terminal embeddings of  ``$\ell_2^2$ distance'' into $\ell_1$. 

%% file: algorithms/algorithm-l1.tex
\begin{figure}[tb]
\begin{framed}
\begin{algorithmic}
\STATE {\bfseries Input:} a data set $X \subset \bbR^d$ and set of centers 
$C=\{c^1,c^2,\dots, c^k\} \subset \bbR^d$
\STATE {\bfseries Output:} a threshold tree $T$
\STATE {}
\STATE Set $S_{ij} = \{\omega\in \Omega: \delta_{c^i}(\omega) \neq \delta_{c^j}(\omega)\}$ for all $i,j \in \{1,\cdots,k\}$. Let $t=0$. 
\STATE Create a tree $T_0$ containing a root vertex $r$. Assign set $X_r = X\cup C$ to the root. 
\STATE {}
\WHILE{$T_t$ contains a leaf with at least two distinct centers $c^i$ and $c^j$}
\STATE Let $E_{t} = \bigcup_{\text{leaves } u}\{(i,j): c^i,c^j\in X_u\}$ be the set of all not yet separated pairs of centers.
\STATE Let $D_t = \max_{(i,j)\in E_t} \|c^i - c^j\|_1$ be the maximum distance between two not separated centers.
\STATE {}
\STATE Define two sets $A_t,B_t\subset \Omega$ as follows:
$$A_t = \bigcup_{(i,j)\in E_t} S_{ij}
\text{\;\;\;\;\;\;\;\;\;\; and\;\;\;\;\;\;\;\;\;\;}
B_t = \bigcup_{\substack{(i,j)\in E_t \\ \text{s.t.} \mu(S_{ij}) \leq D_t/k^3}} S_{ij}.
$$
\STATE Let\footnotemark $R_t = A_t \setminus B_t$. Pick a pair $\omega_t = (i,\theta)$ uniformly at random from $R_t$.
\STATE For every leaf node $u$ in $T$, split the set $X_u$ into two sets: 
$$
\text{\emph{Left}} = \{x\in X_u: x_i \leq \theta\}
\text{\;\;\;\;\;\;\;\;\;\; and\;\;\;\;\;\;\;\;\;\;}
\text{\emph{Right}} = \{x\in X_u: x_i > \theta\}.$$ 
If each of these sets contains at least one center from $C$, then create two children of $u$ in tree $T$ and assign sets \emph{Left} and \emph{Right} to the left and right child, respectively.
\STATE {}
\STATE Denote the updated tree by $T_{t+1}$.
\STATE Update $t = t+1$.
\ENDWHILE
\end{algorithmic}
\end{framed}
   \caption{Threshold tree construction for $k$-medians in $\ell_1$}
   \label{alg:threshold_tree_kmedians}   
\end{figure}

%% file: paper/proof-overview.tex
\section{Algorithms Overview}
We now give an overview of our algorithms.

\medskip

\noindent{\textbf{$k$-medians in $\ell_1$.}}  We begin with the algorithm for $k$-medians in $\ell_1$. We show that its competitive ratio is $O(\log^2 k)$ in Section~\ref{sec:k-medians} and then show an improved bound of $O(\log k\log \log k)$ in Section~\ref{sec:k-medians-improved}.

As the algorithm by~\citet{DFMR20}, our algorithm (see Algorithm~\ref{alg:threshold_tree_kmedians}) builds a binary threshold tree $T$ top-down. It starts with a tree containing only the root node $r$. This node is assigned the set of points $X_r$ that contains all points in the data set $X$ and all reference centers $c^i$. At every round, the algorithm picks some pair $\omega=(i,\theta)\in \Omega$ 
(as we discuss below) and then splits 
data points $x$ assigned to every \emph{leaf} node $u$ into two groups 
$\{x\in X_u: x_i \leq \theta\}$ and $\{x\in X_u:\ x_i > \theta\}$.
Here, $X_u$ denotes the set of points assigned to the node $u$. If this partition separates at least two centers $c^i$ and $c^j$, then the algorithm attaches two children to $u$ and assigns the first group to the left child and the second group to the right child. The algorithm terminates when all leaves contain exactly one reference center $c^i$. Then, we assign the points in each leaf of $T$ to its unique reference center. Note that the unique reference center in each leaf may not be the optimal center for points contained in that leaf. Thus, the total cost by assigning each point to the reference center in the same leaf of $T$ is an upper bound of the cost of threshold tree $T$. 

\footnotetext{As we discuss in Section~\ref{sec:fast_alg}, we can also let $R_t=A_t$. However, this change will make the analysis of the algorithm a little more involved.}

The algorithm by~\citet{DFMR20} picks splitting cuts in a greedy way. Our algorithm chooses them at random. 
Specifically, to pick a cut $\omega_t \in \Omega$ at round $t$, our algorithm finds the maximum distance $D_t$ between two distinct centers $c^i$, $c^j$ that belong to the same set $X_u$ assigned to a leaf node $u$ i.e., 
$$D_t = \max_{u \text{ is a leaf}} \max_{c^i,c^j \in X_u}\|c^i - c^j\|_1.$$
Then, we let $A_t$  be the set of all $\omega\in\Omega$ that separate at least one pair of centers; and $B_t$ be the set of all $\omega\in\Omega$ that separate two centers at distance at most $D_t/k^3$. We pick $\omega_t$ uniformly at random (with respect to measure $\mu$) from the set $R_t = A_t \setminus B_t$. 

Every $\omega \in R_t$ is contained in $A_t$, which means $\omega$ separates at least one pair of centers. Thus, our algorithm terminates in at most $k-1$ iterations. It is easy to see that the running time of this algorithm is polynomial in the number of clusters $k$ and dimension of the space $d$. In Section~\ref{sec:fast_alg}, we provide a variant of this algorithm with 
running time $\tilde O(kd)$.  

\noindent{\textbf{$k$-medians in $\ell_2$.}} Our algorithm for $k$-medians with $\ell_2$ norm recursively partitions the data set $X$ using the following idea. It finds the median point $m$ of all centers in $X$. Then, it repeatedly makes cuts that separate centers from $m$. To make a cut, the algorithm chooses a random coordinate $i \in \{1,\dots,d\}$, random number $\theta\in[0,R^2]$, and random sign $\sigma\in\{\pm 1\}$, where $R$ is the largest distance from a center in $X$  to the median point $m$. It then makes a threshold cut $(i,m _i + \sigma \sqrt{\theta})$. After separating more than half centers from $m$, the algorithm recursively calls itself for each of the obtained parts. In Section~\ref{sec:L2-k-medians}, we show that the \emph{price of explainability} for this algorithm is $O(\log^{\nicefrac32}k$).

\noindent{\textbf{$k$-means.}} We now move to the algorithm for $k$-means. This algorithm embeds the space $\ell_2$ into $\ell_1$ using a specially crafted \emph{terminal embedding} $\varphi$ (the notion of terminal embeddings was formally defined by~\citet{EFN17}). The embedding satisfies the following property for every center $c$ (terminal) and every point $x\in \ell_2$, we have
$$\|\varphi(x) - \varphi(c)\|_1\leq \|x-c\|_2^2 \leq 8k\cdot \|\varphi(x) - \varphi(c)\|_1.$$
Then, the algorithm partitions the data set $\varphi(X)$ with centers $\varphi(c^1), \dots, \varphi(c^k)$ using Algorithm~\ref{alg:threshold_tree_kmedians}. 
The expected cost of partitioning  is at most the distortion of the embedding ($8k$) times the competitive guarantee ($O(\log k\log \log k)$) of Algorithm~\ref{alg:threshold_tree_kmedians}.
In Section~\ref{sec:lb-kmeans}, we show an almost matching lower bound of $\Omega(k/\log k)$ on the cost of explainability for $k$-means. We also remark that the terminal embedding we use in this algorithm cannot be improved. This follows from the fact that the cost function $\|x-c\|_2^2$ does not satisfy the triangle inequality; while the $\ell_1$ distance $\|\varphi(x)-\varphi(c)\|_1$ does. 

%% file: paper/L1-k-medians.tex
\section{\texorpdfstring{Algorithm for $k$-medians in $\ell_1$}{Algorithms for k-medians in l1}}\label{sec:k-medians}
In this section, we analyse Algorithm~\ref{alg:threshold_tree_kmedians} for $k$-medians in $\ell_1$ and show that it provides an explainable 
clustering with cost at most $O(\log^2 k)$ times the original cost. We improve this bound to $O(\log k\log\log k)$ in Section~\ref{sec:k-medians-improved}.

Recall, all centers in $C$ are separated by the tree $T$ returned by the algorithm, and each leaf of $T$ contains exactly one center from $C$. For each point $x \in X$, we define its cost in the threshold tree $T$ as 
$$
\alg_{\ell_1}(x) = \norm{x-c}_1,
$$
where $c$ is the center in the same leaf in $T$ as $x$. Then, $\cost_{\ell_1}(X,T) \leq \sum_{x\in X} \alg_{\ell_1}(x)$ (note that the original centers $c^1,\dots, c^k$ used in the definition of $\alg_{\ell_1}(x)$ are not necessarily optimal for the tree $T$. Hence, the left hand side is not always equal to the right hand side.).
For every point $x \in X$, we also define 
$\mathrm{cost}_{\ell_1}(x,C) = \min_{c\in C} \|x-c\|_1$. Then, $\mathrm{cost}_{\ell_1}(X,C) = \sum_{x \in X} \mathrm{cost}_{\ell_1}(x,C)$
(see~(\ref{def:cost1})).

We prove the following theorem.
\begin{theorem}\label{thm:functional-L1-kmedians}
Given a set of points $X$ in $\bbR^d$ and a set of centers $C = \{c^1,\dots,c^k\}\subset \bbR^d$, Algorithm~\ref{alg:threshold_tree_kmedians} finds a threshold tree $T$ with expected $k$-medians in $\ell_1$ cost at most
$$
\E[\cost_{\ell_1}(X,T)] \leq O(\log^2 k)\cdot \cost_{\ell_1}(X,C).
$$
Moreover, the same bound holds for the cost of every point $x\in X$ i.e.,
$$
\E[\cost_{\ell_1}(x,T)] \leq O(\log^2 k)\cdot \cost_{\ell_1}(x,C).
$$
\end{theorem}
\begin{proof}

Let $T_t$ be the threshold tree constructed by Algorithm~\ref{alg:threshold_tree_kmedians} before iteration $t$. Consider a point $x$ in $X$. 
If $x$ is separated from its original center in $C$ by the cut generated at iteration $t$, then $x$ will be eventually assigned to
some other center in the same leaf of $T_t$. By the triangle inequality, the new cost of $x$ at the end of the algorithm will be
at most $\cost_{\ell_1}(x,C)+D_t$, where $D_t$ is the maximum diameter of any leaf in $T_t$ (see Algorithm~\ref{alg:threshold_tree_kmedians}). 
Define a penalty function $\phi_t(x)$ as follows: $\phi_t(x) = D_t$ if $x$ is separated from its original center $c$ at time $t$; $\phi_t(x)=0$, otherwise. Note that $\phi_t(x) \neq 0$ for at most one iteration $t$, and 
\begin{equation}\label{eq:l1-penalty}
\alg_{\ell_1}(x) \leq \cost_{\ell_1}(x,C) + \sum_{t} \phi_t(x).    
\end{equation}
The sum in the right hand side is over all iterations of the algorithm. We bound the expected penalty $\phi_t(x)$ for each~$t$.

\begin{lemma}\label{lem:recursive}
The expected penalty $\phi_{t}(x)$ is upper bounded as follows:
$$
\E[\phi_{t}(x)] \leq \E\bigg[D_t \cdot \int_{\Omega} 
|\delta_x(\omega) - \delta_c(\omega)|\cdot \frac{\ind\{\omega\in R_t\}}{\mu(R_t)} \,d\mu(\omega)\bigg],
$$
where $c$ is the closest center to the point $x$ in $C$; $\ind\{\omega\in R_t\}$ is the indicator of the event 
$\omega \in R_t$.
\end{lemma}
\begin{proof}
If $x$ is already separated from its original center $c$ at iteration $t$, then $\phi_t(x)=0$. Otherwise, $x$ and $c$ are separated at iteration $t$ if for the random pair 
$\omega_t = (i, \theta)$ chosen from $R_t$ in Algorithm~\ref{alg:threshold_tree_kmedians}, we have $\delta_x(\omega_t)\neq \delta_c(\omega_t)$. Write,
$$\E[\phi_t(x)]\leq\E\Big[\pr_{\omega_t}[\delta_x(\omega_t)\neq \delta_c(\omega_t)\mid T_t]\cdot D_t\Big].$$
The probability that $\delta_x(\omega_t)\neq \delta_c(\omega_t)$ given $T_t$ is bounded as
\begin{align*}
\pr_{\omega_t}[\delta_x(\omega_t&)\neq \delta_c(\omega_t)\mid T_t] = 
\frac{\mu\{\omega\in R_t: \delta_x(\omega) \neq \delta_c(\omega)\}}{\mu(R_t)}
\\
&=\int_{\Omega} \ind\{\delta_x(\omega) \neq \delta_c(\omega)\}\cdot \frac{\ind\{\omega\in R_t\}}{\mu(R_t)} \,d\mu(\omega)
\\
&=\int_{\Omega} |\delta_x(\omega) - \delta_c(\omega)|\cdot \frac{\ind\{\omega\in R_t\}}{\mu(R_t)} \,d\mu(\omega).
\end{align*}
\end{proof}
Let 
$$W_t(\omega) = \frac{D_t\cdot \ind\{\omega\in R_t\}}{\mu(R_t)} .$$
Then, by Lemma~\ref{lem:recursive} and inequality~(\ref{eq:l1-penalty}), we have
$$
\E[\alg_{\ell_1}(x)] \leq \cost_{\ell_1}(x,C) 
+ \E\Big[\sum_{t} \int_{\Omega}  |\delta_x(\omega)- \delta_c(\omega)|\cdot W_t(\omega) \,d\mu(\omega)\Big].
$$
The upper bound on the expected cost of $x$ in tree $T$ consists of two terms: The first term is the original cost of $x$. The second term is a bound on the expected penalty incurred 
by~$x$. We now bound the second term  as $O(\log^2 k)\cdot \cost_{\ell_1}(x,C)$.
$$
\E\Big[\sum_{t} \int_{\Omega}  |\delta_x(\omega)- \delta_c(\omega)|\cdot W_t(\omega) \,d\mu(\omega)\Big]
=
\int_{\Omega}  |\delta_x(\omega)- \delta_c(\omega)|\cdot \E\Big[\sum_{t} W_t(\omega) \Big]
\,d\mu(\omega).
$$

By H\"older's inequality, the right hand side is upper bounded by the following product:
$$
\|\delta_x - \delta_c\|_1
\cdot\max_{\omega\in \Omega} 
\E\Big[\sum_{t} W_t(\omega)\Big].
$$
The first multiplier in the product exactly equals $\|x-c\|_1$
(see~Equation~\ref{eq:isometry}), which, in turn, equals $\cost_{\ell_1}(x,C)$. Hence, to finish the proof of Theorem~\ref{thm:functional-L1-kmedians}, we need to upper bound the second multiplier by $O(\log^2 k)$.

\begin{lemma}\label{lem:weight}
For all $\omega \in \Omega$, we have
$$
\E\Big[\sum_{t} W_t(\omega)\Big] \leq O(\log^2 k).
$$
\end{lemma}
\noindent \textit{Proof.}
Let $t'$ be the first iteration and $t''$ be the last iteration for which $W_t(\omega) > 0$. First, we prove that $D_{t''}\geq D_{t'}/k^3$, where $D_{t'}$ and $D_{t''}$ are the maximum cluster diameters at iterations $t'$ and $t''$, respectively. Since $W_{t'}(\omega) > 0$ and $W_{t''}(\omega) > 0$, we have $\ind\{\omega\in R_{t'}\} \neq 0$ and $\ind\{\omega\in R_{t''}\} \neq  0$. Hence, $\omega\in R_{t'}$ and $\omega\in R_{t''}$. Since $\omega\in R_{t''}$, there exists a pair $(i,j)\in E_{t''}$ for which $\omega\in S_{ij}$. For this pair, we have $D_{t''}\geq \mu(S_{ij})$.
Observe that the pair $(i,j)$ also belongs to $E_{t'}$, since $E_{t''}\subset E_{t'}$. Moreover, $\mu(S_{ij}) > D_{t'}/k^3$, because otherwise, $S_{ij}$ would be included in $B_{t'}$ (see Algorithm~\ref{alg:threshold_tree_kmedians}) and, consequently, $\omega$
would not belong to $R_{t'}=A_{t'}\setminus B_{t'}$. Thus
\begin{equation}\label{eq:D-D-prime}
D_{t''}\geq \mu(S_{ij})> D_{t'}/k^3.
\end{equation}

By the definition of $t'$ and $t''$, we have
$$
\sum_{t} W_t(\omega) = \sum_{t=t'}^{t''} W_t(\omega) \leq \sum_{t=t'}^{t''} \frac{D_{t}}{\mu(R_t)}.
$$
Note that the largest distance $D_t$ is a non-increasing (random) function of $t$. Thus, 
we can split the iterations of the algorithm $\{t',...,t''\}$ into $\ceil{3\log k}$ phases. At phase $s$, the maximum diameter $D_t$ is in the range
$(D_{t'}/2^{s+1},D_{t'}/2^s]$.
Denote the set of all iterations in phase $s$ by $\Phase(s)$.

Consider phase $s$. Let $D = D_{t'}/2^{s}$. Phase $s$ ends when all sets $S_{ij}$ with $\mu(S_{ij})\geq D/2$ are removed from the set $E_t$. Let us estimate the probability that one such set $S_{ij}$ is removed from $E_t$ at iteration $t$. Set $S_{ij}$ is removed from $E_t$ if the random threshold cut $\omega_t$ chosen at iteration $t$
separates centers $c_i$ and $c_j$, or, in other words, if $w_t\in S_{ij}$. The probability of this event equals:
$$
\pr[\omega_t \in S_{ij}\mid T_t]  = \frac{\mu(S_{ij}\cap R_t)}{\mu(R_t)}
= \frac{\mu(S_{ij})-\mu(S_{ij}\cap B_t)}{\mu(R_t)}\geq \frac{\mu(S_{ij})-\mu(B_t)}{\mu(R_t)}.
$$
Note that $\mu(S_{ij}) > D/2\geq D_t/2$ and $\mu(B_t)< \binom{k}{2} \cdot \frac{D_t}{k^3} < \frac{D_t}{2k}$ (because $B_t$ is the union of at most $\binom{k}{2}$ sets of measure at most $D_t/k^3$ each). Hence, 
$$\pr[\omega_t \in S_{ij}\mid T_t]
\geq \frac{D_t}{4\mu(R_t)}
\geq \frac{1}{4} W_t(\omega).
$$
If $W_t(\omega)$ did not depend on $t$, then we would argue that each set $S_{ij}$ (with $\mu(S_{ij})\geq D/2$) is removed from $E_t$ in 
at most $4/W_t(\omega)$ iterations, in expectation, and, consequently, all sets $S_{ij}$ are removed in at most $O(\log k) \cdot 4/W_t(\omega)$ iterations, in expectation (note that the number of sets $S_{ij}$ is upper bounded by $\binom{k}{2}$). Therefore, 
$$
\E\Big[\sum_{t\in \Phase(s)} W_t(\omega)\Big] \leq 
O(\log k) \cdot \frac{4}{W_t(\omega)}\cdot W_t(\omega)
= O(\log k).
$$
However, we cannot assume that $W_t(\omega)$ is a constant. Instead, we use the following claim with 
$E=\{0,\dots,k-1\}\times\{0,\dots,k-1\}$, $E'_t=\{(i,j)\in E_t: \mu(S_{ij})\geq D/2\}$,
and $p_t = W_t(\omega)/4$.

\begin{claim}\label{claim:sum-probs}
Consider two stochastic processes $E_t$ and $p_t$ adapted to filtration $\calF_t$. The values of $E_t$ are subsets of some finite non-empty set $E$. The values of $p_t$ are numbers in $[0,1]$. Suppose that for every step $t$, $E_{t+1} \subset E_t$ and for every $e\in E_t$, $\Pr[e\notin E_{t+1}\mid \calF_t] \geq p_t$. Let 
$\tau$ be the (stopping) time $t$ when $E_t=\varnothing$.
Then,
$$\E\bigg[\sum_{t=0}^{\tau - 1} p_t \bigg]\leq \ln |E| + O(1).$$
\end{claim}

\begin{proof}[Proof of Claim~\ref{claim:sum-probs}]
Let $\tau_e$ be the first time $t\geq 1$ when element $e$ does not belong to $E_t$. Then, $\tau=\max_{e\in E} \tau_e$. Hence,
$$\sum_{t=0}^{\tau - 1} p_t = \max_{e\in E} \sum_{t=0}^{\tau_e - 1} p_t.$$
By the union bound, for all $\lambda \geq 0$, we have 
\begin{equation}\label{eq:u-bound}
\pr\bigg[\sum_{t=0}^{\tau - 1} p_t \geq \lambda\bigg]
\leq
\sum_{e\in E}\pr\bigg[\sum_{t=0}^{\tau_e - 1} p_t \geq \lambda\bigg].
\end{equation}
Define a new stochastic process $Z_t(e)$ as follows: $Z_0(e)=1$ and for $t\geq 1$,
$$Z_t(e) = 
\begin{cases}
e^{\sum_{t'=0}^{t-1} p_{t'}},&\text{if } e\in E_t;\\
0,&\text{otherwise}.
\end{cases}$$
Note that if $\sum_{t=0}^{\tau_e - 1} p_t\geq \lambda$, then $\max_{t\geq 0} Z_t(e)\geq e^{\lambda - 1}$. Thus, we will bound $\Pr[\max_{t\geq 0} Z_{t}(e)\geq e^{\lambda-1}]$.
Observe that $Z_t$ is a supermartingale, since 
\begin{align*}
\E[Z_{t+1} \mid \calF_t] &= \Pr[e\in E_{t+1}\mid \calF_t]\cdot e^{p_{t}}\cdot Z_t 
\\ &\leq (1-p_t)\cdot e^{p_{t}}\cdot Z_t\leq Z_t.
\end{align*}
By Doob's maximal martingale inequality, we have 
$$
\Pr[\max_{t\geq 0} Z_{t}(e)\geq e^{\lambda-1}] \leq 
Z_0(e)/ e^{\lambda-1} = e^{-(\lambda-1)}.
$$
Using~(\ref{eq:u-bound}), we get 
$$
\pr\bigg[\sum_{t=0}^{\tau - 1} p_t \geq \lambda\bigg]
\leq |E|\cdot e^{-(\lambda-1)}.
$$
Therefore,
\begin{align*}
\E\bigg[\sum_{t=0}^{\tau - 1} p_t\bigg] &= 
\int_0^\infty \pr\bigg[\sum_{t=0}^{\tau - 1} p_t \geq \lambda\bigg] d\lambda \leq \ln |E| + \int_{\ln |E|}^\infty |E|\cdot e^{-(\lambda-1)} d\lambda\\
&= \ln |E| + |E|e^{-\ln |E| + 1} = \ln|E|+e.
\end{align*}
\end{proof}

By Claim~\ref{claim:sum-probs}, 
$$\E\Big[\sum_{t\in \Phase(s)} W_t(\omega)\Big] \leq O(\log k).$$
The expected sum of $W_t$ over all phases is upper bounded by $O(\log^2 k)$, since the number of phases is upper bounded by $O(\log k)$. We note that if the number of phases is upper bounded by $L$, then the expected sum of $W_t$ over all phases is upper bounded by $O(L \log k)$. This concludes the proofs of Lemma~\ref{lem:weight}
and Theorem~\ref{thm:functional-L1-kmedians}.
\end{proof}

%% file: paper/L1-improved.tex
\section{\texorpdfstring{Improved Analysis for $k$-medians in $\ell_1$}{Improved Analysis for k-medians in l1}}\label{sec:k-medians-improved}
In this section, we provide an improved analysis of our algorithm for $k$-medians in $\ell_1$. 

\begin{theorem}\label{thm:improved}
Given a set of points $X$ in $\bbR^d$ and set of centers $C = \{c^1,\dots,c^k\}\subset \bbR^d$, Algorithm~\ref{alg:threshold_tree_kmedians} finds a threshold tree $T$ with expected $k$-medians $\ell_1$ cost at most
$$
\E[\cost_{\ell_1}(X,T)] \leq O(\log k\log \log k)\cdot \cost_{\ell_1}(X,C).
$$
\end{theorem}
\begin{proof}
In the proof of Theorem~\ref{thm:functional-L1-kmedians}, we used  a pessimistic estimate on the penalty a point $x\in X$ incurs when it is separated from its original center $c$. Specifically, we bounded the penalty by the maximum diameter of any leaf in the tree $T_t$. In the current proof, we will use an additional bound: The distance from $x$ to the closest center after separation. Suppose, that $x$ is separated from its original center $c$. Let $c'$ be the closest center to $x$ after we make cut $\omega_t$ at step $t$. That is, $c'$ is the closest center to $x$ in the same leaf of the threshold tree $T_{t+1}$. Note that after we make additional cuts, $x$ may be separated from its new center $c'$ as well, and the cost of $x$ may increase. However, as we already know, the expected cost of $x$ may increase in at most $O(\log^2 k)$ times in expectation (by Theorem~\ref{thm:functional-L1-kmedians}). Here, we formally apply Theorem~\ref{thm:functional-L1-kmedians} to the leaf where $x$ is located and treat $c'$ as the original center of $x$. Therefore, if $x$ is separated from $c$ by a cut $\omega_t$ at step $t$, then the expected cost of $x$ in the end of the algorithm is upper bounded by
\begin{equation}
\E[\alg_{\ell_1}(x)\mid T_t,\omega_t] \leq O(\log^2 k)\cdot \|c'-x\|_1
= O(\log^2 k)\cdot D^{min}_t(x,\omega_t).\nonumber
\end{equation}
In the formula above, we used the following definition: $D^{min}_t(x,\omega)$ is the distance from $x$ to the closest center $c'$ in the same leaf of $T_t$ as $x$ which is not separated from $x$ by the cut $\omega$ i.e.,
$\delta_x(\omega)=\delta_{c'}(\omega)$. 
If there are no such centers $c'$ (i.e., cut $\omega$ separates $x$ from all centers), then we let $D^{min}_t(x,\omega) = 0$. Note that in this case, our algorithm will never make cut $\omega$, since it always makes sure that the both parts of the cut contain at least one center from $C$. Similarly to $D^{min}_t(x,\omega)$, we define $D^{max}_t(x,\omega)$: $D^{max}_t(x,\omega)$ is the distance from $x$ to the farthest center $c''$ in the same leaf of $T_t$ as $x$ which is not separated from $x$ by the cut $\omega$. We also let  $D^{max}_t(x,\omega) = 0$ if there is no such $c''$. Note that $D^{max}_t(x,\omega)$ is an upper bound on the cost of $x$ in the eventual threshold tree $T$ if cut $\omega$ separated $x$ from $c$ at step~$t$.

We now have three bounds on the expected cost of $x$ in the final tree $T$ given that the algorithm separates $x$ from its original center $c$ at step $t$ with cut $\omega$. 
The first bound is $D_t^{max}(x,\omega)$; the second bound is  $O(\log^2 k)\cdot D^{min}_t(x,\omega)$, and the third bound is $\|x-c\|_1+D_t$. We use the first bound if $D^{max}_t(x,\omega)\leq 2\|x-c\|_1$. We call such cuts $\omega$ light cuts. We use the second bound if $D^{max}_t(x,\omega) > 2\|x-c\|_1$ but $D^{min}_t(x,\omega)\leq D_t/\log^4 k$. We call such cuts $\omega$ medium cuts. We use the third bound if $D^{max}_t(x,\omega)> 2\|x-c\|_1$ and $D^{min}_t(x,\omega)> D_t/\log^4 k$. We call such cuts $\omega$ heavy cuts.

Note that in the threshold tree returned by the algorithm, one and only one of the following may occur: (1) $x$ is separated from the original center $c$ by a light, medium, or heavy cut; (2) $x$ is not separated from $c$.  We now estimate expected penalties due to light, medium, or heavy cuts.

If the algorithm makes a light cut, then the maximum cost of point $x$ in $T$ is at most $2\|x-c\|_1 = 2\cost_{\ell_1}(x,C)$. So we should not worry about such cuts. If the algorithm makes a medium cut, then the expected additional penalty for $x$ is upper bounded by 
$$D^{min}_t(x,\omega_t)\cdot O(\log^2 k) \leq O(\phi_t(x)/\log^2 k),$$
where $\phi_t(x)$ is the function from the proof of Theorem~\ref{thm:functional-L1-kmedians}. Thus, the total expected penalty due to a medium cut (added up over all steps of the algorithm)  is $\Omega(\log^2 k)$ times smaller than the penalty we computed in the proof of Theorem~\ref{thm:functional-L1-kmedians}. Therefore, the expected penalty due to a medium cut is at most $O(\|x-c\|_1)$.

We now move to heavy cuts. Denote the set of possible heavy cuts for $x$ in $R_t$ by $H_t$. That is, if $x$ is not separated from its original center $c$ by step $t$, then
$$
H_t=\Big\{\omega\in R_t:\;  D^{min}_t(x,\omega) > D_t/\log^4 k \text{ and } 
D^{max}_t(x,\omega) > 2\|x-c\|_1\Big\}.
$$
Otherwise, let $H_t = \varnothing$.
Define a density function $\widetilde W_t(\omega)$ similarly to 
$W_t(\omega)$:
$$\widetilde W_t(\omega) = \frac{D_t\cdot \ind\{\omega \in H_t\}}{\mu(R_t)} .$$
Then, the expected penalty due to a heavy cut is bounded, similarly to Lemma~\ref{lem:recursive},
by 
$$
\sum_t \E\Bigg[\int_0^1 |\delta_x(\omega)-\delta_c(\omega)|\cdot \widetilde W_t \,d\mu(\omega)\Bigg].
$$
Therefore, to finish the proof of Theorem~\ref{thm:kmedians-improved}, we need to prove the following analog of Lemma~\ref{lem:weight}.
\begin{lemma}\label{lem:tilde-weight}
For all $\omega \in \Omega$, we have
$$
\E\Big[\sum_{t} \widetilde W_t(\omega)\Big] \leq O(\log k\log \log k).
$$
\end{lemma}
\noindent\textit{Proof.}
As in the proof of Lemma~\ref{lem:weight}, consider the first and last steps when 
$\widetilde W_t(\omega) > 0$. Denote these steps by $t^{*}$ and $t^{**}$, respectively. In the proof  of Lemma~\ref{lem:weight}, we had a bound $D_{t''}\geq D_{t'}/k^3$ (see inequality~(\ref{eq:D-D-prime})). We now show a stronger bound on $t^{*}$ and $t^{**}$.
\begin{claim}\label{cl:D-D-star} We have
$D_{t^{**}} \geq  \nicefrac{D_{t^*}}{2\log^4 k}$.
\end{claim}
This claim implies that the number of phases defined in Lemma~\ref{lem:weight} is bounded by $O(\log \log k)$, which immediately implies Lemma~\ref{lem:tilde-weight}. So, to complete the proof, it remains to show Claim~\ref{cl:D-D-star}.

\noindent\textit{Proof of Claim~\ref{cl:D-D-star}}
First, note that $\ind\{\omega \in H_{t^{**}}\}> 0$ and, consequently, cut $\omega$ is heavy at step $t^{**}$. Thus, $D^{min}_{t^{**}}(x,\omega)$ is positive. Hence, this cut separates $c$ from at least one other center $c'$ in the same leaf of the current threshold tree $T_{t^{**}}$. Let $c''$ be the farthest such center from point $x$. Then, $\|c''-x\|_1 = D^{max}_{t^{**}}(x,\omega)$. Since centers $c$ and $c''$ are not separated prior to step $t^{**}$, we have
$$D_{t^{**}}\geq \|c-c''\|_1
\geq \|x-c''\|_1 - \|x-c\|_1.$$
Since $\omega$ is a heavy cut and not a light cut, $\|x-c''\|_1 > 2 \|x-c\|_1$. Thus,
$$D_{t^{**}}\geq \frac{\|x-c''\|_1}{2}
= 
\frac{D^{max}_{t^{**}}(x,\omega)}{2} \geq 
\frac{D^{min}_{t^{**}}(x,\omega)}{2}.$$
Now, observe that the  random process $D^{min}_{t^{**}}(x,\omega)$ is non-decreasing (for fixed $x$ and $\omega$) since the distance from $x$ to the closest center $c'$ cannot decrease over time. Therefore,
$$D_{t^{**}}\geq \frac{D^{min}_{t^{**}}(x,\omega)}{2}\geq
\frac{D^{min}_{t^{*}}(x,\omega)}{2}\geq \frac{D_{t^{*}}}{2\log^4 k}.$$
In the last inequality, we used that $\omega$ is a heavy cut at time $t^{*}$. This finishes the proof of Claim~\ref{cl:D-D-star}.
\end{proof}

%% file: paper/terminal-embed-L2L1.tex
\section{\texorpdfstring{Terminal Embedding of $\ell_2^2$ into $\ell_1$}{Terminal Embedding of squared l2 into l1}}

In this section, we show how to construct a coordinate cut preserving terminal embedding of $\ell_2^2$ (squared Euclidean distances) into $\ell_1$ with distortion $O(k)$ for every set of terminals $K \subset \bbR^d$ of size $k$.

Let $K$ be a finite subset of points in $\bbR^d$. We say that $\varphi: x \mapsto \varphi(x)$ is a terminal embedding of $\ell_2^2$ into $\ell_1$ with a set of terminals $K$ and distortion $\alpha$ if for every terminal $y$ in $K$ and every point $x$ in $\bbR^d$, we have 
$$\|\varphi(x) - \varphi(y)\|_1 \leq \|x-y\|_2^2 \leq \alpha \cdot \|\varphi(x) - \varphi(y)\|_1. $$

\begin{lemma}\label{lem:embedding}
For every finite set of terminals $K$ in $\bbR^d$, there exists a coordinate cut preserving terminal embedding of $\ell_2^2$ into $\ell_1$
with distortion $8|K|$.
\end{lemma}
\begin{proof}
We first prove a one dimensional analog of this theorem (which corresponds to the case when all points and centers are in one dimensional space).
\begin{lemma}\label{lem:terminal-embedding-1d}
For every finite set of real numbers $K$, there exists a cut preserving embedding 
$\psi_K:\bbR\to\bbR$ such that for every $x\in \bbR$ and $y\in K$, we have
\begin{equation}\label{ineq:psi-K}
|\psi_K(x)- \psi_K(y)| \leq
|x-y|^2
 \leq 8|K|\cdot|\psi_K(x)- \psi_K(y)|. 
\end{equation}
\end{lemma}

\begin{proof}
Let $k$ be the size of $K$ and $y_1,\dots,y_k$ be the elements of $K$ sorted in increasing order. We first define $\psi_K$ on points in $K$ and then extend this map to the entire real line $\bbR$. We map each $y_i$ to $z_i$ defined as follows:
$z_1 = 0$ and for $i=2,\dots, k$,
$$z_i = \frac{1}{2}\sum_{j=1}^{i-1} (y_{j+1}-y_j)^2.$$

Now consider an arbitrary number 
$x$ in $\bbR$. Let $y_i$ be the closest point to $x$ in $K$. Let 
$\varepsilon_x = \sign(x- y_i)$. 
Then, $x = y_i + \varepsilon_x |x-y_i|$. Note that $\varepsilon_x = 1$ if $x$ is on the right to $y_i$, and $\varepsilon_x = -1$, otherwise. Let the function $\psi_K$ be
$$\psi_K(x) = z_i + \varepsilon_x (x-y_i)^2.$$
For $x = (y_i+y_{i+1})/2$, both $y_i$ and $y_{i+1}$ are the closest points to $x$ in $K$. In this case, we have 
$$z_i + \varepsilon_x (x-y_i)^2 = z_{i+1} + \varepsilon_x (x-y_{i+1})^2,$$
which means $\psi_K(x)$ is well-defined. 

An example of the terminal embedding function $\psi_K(x)$ is shown in Figure~\ref{fig:example}.  
Then, we show that this function $\psi_K$ is a cut preserving embedding satisfying inequality~(\ref{ineq:psi-K}). 

We first show that this function $\psi_K$ is continuous and differentiable in $\bbR$. Consider $2k$ open intervals on the real line divided by points in $K$ and points $(y_i+y_{i+1})/2$ for $i\in \{1,2,\cdots,k-1\}$. In every such open interval, the function $\psi_K$ is a quadratic function, which is continuous and differentiable. Since $\psi_K$ is also continuous and differentiable at the endpoints of these intervals, the function $\psi_K$ is continuous and differentiable in $\real$. For any $x \in \bbR$, we have $\psi_K'(x) = 2\abs{x-y^*} \geq 0$ where $y^*$ is the closest point in $K$ to $x$. Thus, the function $\psi_K$ is increasing in $\bbR$, which implies $\psi_K$ is cut preserving.  

We now prove that $\psi_K$ satisfies two inequalities. We first show that for every $x\in \bbR$ and $y \in K$, 
$|\psi_K(x)-\psi_K(y)|\leq 
|x-y|^2$. Suppose that $x\geq y$ (The case $x\leq y$ is handled similarly.) If $x=y$, then this inequality clearly holds. Thus, to prove $|\psi_K(x)-\psi_K(y)|\leq |x-y|^2$, it is 
sufficient to prove the following inequality on derivatives 
$$(\psi_K(x)-\psi_K(y))'_x\leq \big((x-y)^2\big)'_x.$$
Let $y^*$ be the closest point in $K$ to $x$. Then,
$$
(\psi_K(x)-\psi_K(y))'_x = (\psi_K(x))'_x 
= (\psi_K(y^*) + \varepsilon_x(x-y^*)^2)'_x 
= 2|x-y^*|.
$$
Since $y^*$ is the closest point in $K$ to $x$, we have 
$|x-y^*|\leq |x-y| = \big((x-y)^2\big)'_x /2$.
This finishes the proof of the first inequality.

\medskip

We now verify the second inequality. First, consider 
two points $y_i$ and $y_j$ ($y_i < y_j$). Write,
$$\psi_K(y_j) - \psi_K(y_i) = z_j - z_i =
\frac{1}{2} \sum_{m=i}^{j-1} (y_{m+1}-y_m)^2.$$
By the arithmetic mean--quadratic mean inequality, we have
$$
(j-i)\cdot \sum_{m=i}^{j-1} (y_{m+1}-y_m)^2 \geq
\Big(\sum_{m=i}^{j-1} y_{m+1}-y_m\Big)^2 = 
(y_j - y_i)^2.
$$
Thus,
$$\psi_K(y_j) - \psi_K(y_i)\geq \frac{(y_j-y_i)^2}{2(j-i)}\geq 
\frac{(y_j-y_i)^2}{2(k-1)}.$$

Now we consider the case when $x$ is an arbitrary real number in $\bbR$ and $y\in K$. Let $y^*$ be the closest point in $K$ to $x$. Then,
$$|x-y|^2 \leq 2|x-y^*|^2 + 2|y^*-y|^2.$$

The first term on the right hand side equals
$4|\psi_K(x)-\psi_K(y^*)|$; the second term is upper bounded by $4(k-1) |\psi_K(y)-\psi_K(y^*)|$. Thus,
$$|x-y|^2 \leq 4|\psi_K(x)-\psi_K(y^*)| +
4(k-1) |\psi_K(y^*)-\psi_K(y)|.$$
Note that $|\psi_K(x)-\psi_K(y^*)|\leq |\psi_K(x)-\psi_K(y)|$ since $y^*$ is the closest point in $K$ to $x$. Also, we have
$$
|\psi_K(y^*)-\psi_K(y)|\leq |\psi_K(x)-\psi_K(y^*)| + |\psi_K(x)-\psi_K(y)|\leq 2|\psi_K(x)-\psi_K(y)|.
$$
Hence,
$$|x-y|^2 \leq 8k|\psi_K(x)-\psi_K(y)|.$$
This completes the proof.
\end{proof}


\begin{figure}
    \centering
        \input{figures/embedding}
    \caption{Terminal embedding function $\psi_K(x)$ for $K= \{1,3,5\}$.}
    \label{fig:example}
\end{figure}
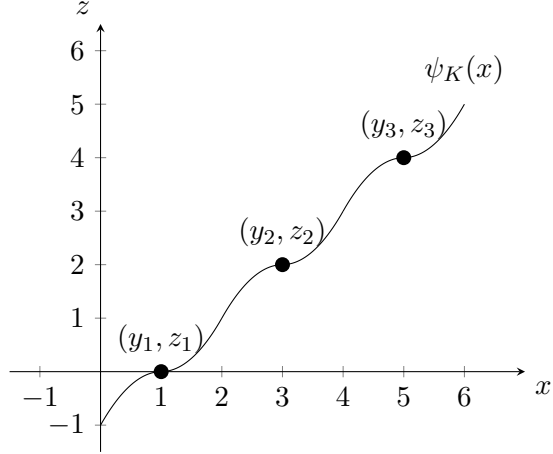

Using the above lemma, we can construct a terminal embedding  $\psi$ from $d$-dimensional $\ell_2^2$ into $d$-dimensional $\ell_1$ as follows. For each coordinate $i\in\{1,2,\cdots,d\}$, let $K_i$ be the set of the $i$-th coordinates for all terminals in $K$. Define one dimensional terminal embeddings $\psi_{i}$ for all coordinates $i$. Then, $\psi$ maps every point $x\in \ell_2^2$ to 
$\psi(x) = (\psi_1(x),\cdots, \psi_d(x))$. 

We show that this terminal embedding $\psi$ is coordinate cut preserving. 
By the construction of $\varphi$, we have for any threshold cut $(i,\theta)$
$$
\{x\in \bbR^d: \psi(x)_i \leq \theta\} = \{x\in \bbR^d: \psi_i(x_i) \leq \theta\}.
$$
Since $\psi_{i}$ is a cut preserving terminal embedding by Lemma~\ref{lem:terminal-embedding-1d}, there exists a threshold $\theta' \in \bbR$ such that
$$
\{x\in \bbR^d: x_i \leq \theta'\} = \{x\in \bbR^d: \psi_{i}(x_i) \leq \theta\},
$$
which implies $\psi$ is coordinate cut preserving.
\end{proof}

For explainable $k$-means clustering, we first use the terminal embedding of $\ell_2^2$ into $\ell_1$. Then, we apply Algorithm~\ref{alg:threshold_tree_kmedians} to the instance after the embedding. By using this terminal embedding, we can get the following result.

\begin{theorem}\label{thm:functional-L2-kmeans}
Given a set of points $X$ in $\bbR^d$ and a set of centers $C$ in $\bbR^d$, Algorithm~\ref{alg:threshold_tree_kmedians} with terminal embedding finds a threshold tree $T$ with expected $k$-means cost at most
$$
\E[\cost_{\ell_2^2}(X,T)] \leq O(k\log k\log\log k)\cdot \cost_{\ell_2^2}(X,C).
$$
\end{theorem}

\begin{proof}
Let $\varphi$ be the terminal embedding of $\ell_2^2$ into $\ell_1$ with terminals $C$. Let $T'$ be the threshold tree returned by our algorithm on the instance after embedding. Since the terminal embedding $\varphi$ is coordinate cut preserving, the threshold tree $T'$ also provides a threshold tree $T$ on the original $k$-means instance. Let $\varphi(C)$ be the set of centers after embedding. For any point $x\in X$, the expected cost of $x$ is at most
\begin{align*}
\E[\mathrm{cost}_{\ell_2^2}(x,T)] &\leq 8k \cdot \E[\mathrm{cost}_{\ell_1}(\varphi(x),T')] \\
&\leq O(k\log k \log\log k)\cdot \mathrm{cost}_{\ell_1}(\varphi(x),\varphi(C)) \\
&\leq O(k\log k \log \log k)\cdot \mathrm{cost}_{\ell_2^2}(x,C),
\end{align*}
where the first and third inequality is from the terminal embedding in Lemma~\ref{lem:embedding} and the second inequality is due to Theorem~\ref{thm:improved}.
\end{proof}

%% file: figures/embedding.tex
\begin{tikzpicture}
\begin{axis}[
  axis x line=center,
  axis y line=center,
  xtick={-1,0,...,5,6},
  ytick={-1,0,...,5,6},
  xlabel={$x$},
  ylabel={$z$},
  xlabel style={below right},
  ylabel style={above left},
  xmin=-1.5,
  xmax=7,
  ymin=-1.5,
  ymax=6.5]
\addplot [mark=none,domain=0:1] {-x*x+2*x-1};
\addplot [mark=none,domain=1:2] {x*x-2*x+1};
\addplot [mark=none,domain=2:3] {-x*x+6*x-7};
\addplot [mark=none,domain=3:4] {x*x-6*x+11};
\addplot [mark=none,domain=4:5] {-x*x+10*x-21};
\addplot [mark=none,domain=5:6] {x*x-10*x+29};
\node[label={{$(y_1,z_1)$}},circle,fill,inner sep=2pt] at (axis cs:1,0) {};
\node[label={{$(y_2,z_2)$}},circle,fill,inner sep=2pt] at (axis cs:3,2) {};
\node[label={{$(y_3,z_3)$}},circle,fill,inner sep=2pt] at (axis cs:5,4) {};
\node[label={{$\psi_K(x)$}}] at (axis cs:6,5) {};
\end{axis}
\end{tikzpicture}

%% file: paper/L2-k-medians.tex
\section{\texorpdfstring{$k$-medians in $\ell_2$}{{k-medians in l2}}}\label{sec:L2-k-medians}

In this section, we present an algorithm for the $k$-medians in $\ell_2$ and show that it provides an explainable clustering with cost at most $O(\log^{3/2} k)$ times the original cost.

\subsection{\texorpdfstring{Algorithm for $k$-medians in $\ell_2$}{Algorithm for k-medians in l2}}
Our algorithm builds a binary threshold tree $T$ using a top-down approach, as shown in Algorithm~\ref{alg:threshold_tree_kmedians-L_2}. It starts with a tree containing only the root node $r$. The root $r$ is assigned the set of points $X_r$ that contains all points in the data set $X$ and all reference centers $c^i$. Then, the algorithm calls function \textsc{Build\_tree($r$)}. Function \textsc{Build\_tree($u$)} partitions centers in $u$ in several groups $X_v$ using function
\textsc{Partition\_Leaf($u$)} and then recursively calls itself (\textsc{Build\_tree($v$)}) for every new group $X_v$ that contains more than one reference center $c^i$.

Most work is done in the function \textsc{Partition\_Leaf($u$)}. The argument of the function is a leaf node $u$ of the tree. We denote the set of data points and centers assigned to $u$ by $X_u$. Function \textsc{Partition\_Leaf($u$)} partitions the set of centers assigned to node $u$ into several groups. Each group contains at most half of all centers $c^i$ from the set $X_u$. When  \textsc{Partition\_Leaf($u$)} is called, the algorithm finds the $\ell_1$-median of all reference centers in node $u$. Denote this point by $m^u$. We remind the reader that the $i$-th coordinate of the median $m^u$ (which we denote by $m^u_i$) is a median for $i$-th coordinates of centers in $X_u$. That is, for each coordinate $i$, both sets $\{c\in X_u\cap C: c_i < m^u_i\}$
and $\{c\in X_u\cap C: c_i > m^u_i\}$ contain at most half of all centers in $X_u$.
Then, function \textsc{Partition\_Leaf($u$)} iteratively partitions 
$X_u$ into pieces until each piece contains at most half of all centers from $X_u$.
We call the piece that contains the median $m^u$ the main part (note that we find the median $m^u$ when \textsc{Partition\_Leaf($u$)} is called and do not update $m^u$ afterwards). 

At every iteration $t$, the algorithm finds the maximum distance $R_t^u$ from centers in the main part to the point $m^u$. The algorithm picks a random coordinate $i_t^u\in \{1,2,\cdots,d\}$, random number $\theta_t^u \in [0, (R_t^u)^2]$, and random sign $\sigma_t^u\in \{\pm 1\}$ uniformly. Then, it splits the main part using the threshold cut $(i_t^u,m^u_i + \sigma_t^u\sqrt{\theta_t^u})$ if this cut separates at least two centers in the main part. Function \textsc{Partition\_Leaf($u$)} stops, when the main part contain at most half of all centers in $X_u$. Note that all pieces separated from $m^u$ during the execution of \textsc{Partition\_Leaf($u$)} contain at most half of all centers in $X_u$ because $m^u$ is the median of all centers in $X_u$.

\input{algorithms/algorithm-L_2}

\begin{theorem}\label{thm:k-medians-l2-first}
Given a set of points $X$ in $\bbR^d$ and a set of centers $C = \{c^1,\dots,c^k\}\subset \bbR^d$, Algorithm~\ref{alg:threshold_tree_kmedians-L_2} finds a threshold tree $T$ with expected $k$-medians in $\ell_2$ cost at most
$$
\E[\cost_{\ell_2}(X,T)] \leq O(\log^{\nicefrac{3}{2}} k)\cdot \cost_{\ell_2}(X,C).
$$
\end{theorem}

\begin{proof}
Let $T_t(u)$ be the threshold tree at the beginning of iteration $t$ in function \textsc{Partition\_Leaf($u$)}.
For every point $x \in X_u$, define its cost at step $t$ of function \textsc{Partition\_Leaf($u$)} to be the distance from $x$ to the closest center in 
the same leaf of $T_t(u)$ as $x$. That is, if $x$ belongs to a leaf node $v$ in the threshold tree $T_t(u)$, then
$$\cost_{\ell_2}(x,T_t(u))= \min\{\|x-c\|_2: c\in X_v \cap C\}.$$
If the point $x$ is separated from its original center in $C$ by the cut generated at time step $t$, then $x$ will be eventually assigned to
some other center in the  main part of $T_t(u)$. 
By the triangle inequality, the new cost of $x$ at the end of the algorithm will be
at most $\cost_{\ell_2}(x,C)+2R_t^u$, where $R_t^u$ is the maximum radius of the main part in $T_t(u)$ i.e., 
$R_t^u$ is the distance from the median $m^u$ to the farthest center $c^i$ in the main part. 
Define a penalty function $\phi_t^u(x)$ as follows: $\phi_t^u(x) = 2R_t^u$ if $x$ is separated from its original center $c$ at time $t$; $\phi_t^u(x)=0$, otherwise. 
Let $U_x$ be the set of all nodes $u$ for which the algorithm calls \textsc{Build\_Tree($u$)} and $x \in X_u$. Note that some nodes $v$ of the threshold tree with 
$x\in X_v$ do not belong to $U_x$. Such nodes $v$ are created and split into two groups in the same call of \textsc{Partition\_Leaf($u$)}. 
Observe that $\phi_t^u(x) \neq 0$ for at most one step $t$ in the call of \textsc{Partition\_Leaf($u$)} for some node $u \in U_x$, and 
\begin{equation}\label{eq:l2-penalty}
\cost_{\ell_2}(x,T) \leq \cost_{\ell_2}(x,C) + \sum_{u\in U_x}\sum_{t} \phi_t^u(x).  
\end{equation}
The sum in the right hand side is over all iterations $t$ in all calls of function \textsc{Partition\_Leaf($u$)} with $u \in U_x$. Since each piece in the partition returned by function \textsc{Partition\_Leaf($u$)} contains at most half of all centers from $X_u$, the depth of the recursion tree is at most $O(\log k)$ (note that the depth of the threshold tree can be as large as $k-1$). This means that the size of $U_x$  is at most $O(\log k)$. In Lemma~\ref{lem:penalty-L_2}, we show that the expected total penalty in the call of \textsc{Partition\_Leaf($u$)} for every $u\in U_x$ is at most $O(\sqrt{\log k})$ times the original cost. Before that, we upper bound the expected penalty $\phi_t^u(x)$ for each step $t$ in the call of \textsc{Partition\_Leaf($u$)} for every node $u \in U_x$. 

\begin{lemma}\label{lem:penalty}
The expected penalty $\phi_t^u(x)$ is upper bounded as follows:
$$
\E[\phi_{t}^u(x)] \leq \E\bigg[2\norm{x-c}_2\cdot\frac{\norm{c - m^u}_2+\norm{x-m^u}_2}{d \cdot R_{t}^u}\bigg],
$$
where $c$ is the closest center to the point $x$ in $C$.
\end{lemma}

\begin{proof}
We first bound the probability that point $x$ is separated from its original center $c$ at iteration $t$. For any coordinate $i \in \{1,2,\cdots,d\}$, let $x_i$ and $c_i$ be the $i$-th coordinates of point $x$ and center $c$ respectively. For any point $x\in \bbR^d$, we define the indicator function  $\delta_x(i,\theta) = 0$ if $x_i \leq \theta$, and $\delta_x(i,\theta) = 1$ otherwise. 
To determine whether the threshold cut sampled at iteration $t$ separates $x$ and $c$, we consider the following two cases: (1) $x$ and $c$ are on the same side of the median $m^u$ in coordinate $i$ (i.e. $(x_i-m^u_i)(c_i - m^u_i) \geq 0$), and (2) $x$ and $c$ are on the opposite sides of the median $m^u$ in coordinate $i$ (i.e. $(x_i-m^u_i)(c_i - m^u_i) < 0$).

If $x$ and $c$ are on the same side of the median $m^u$ in coordinate $i$, then the threshold cut $(i,m^u_i+\sigma_t^u\sqrt{\theta_t^u})$ separates $x$ and $c$ if and only if $\sigma_t^u$ has the same sign as $x_i-m^u_i$ and $\theta_t^u$ is between $(x_i-m^u_i)^2$ and $(c_i-m^u_i)^2$. Thus,
\begin{align*}
    \prob{\delta_x(i,\vartheta_t^u) \neq \delta_c(i,\vartheta_t^u) \mid T_t(u)} &=\frac{\abs{(c_i-m^u_i)^2-(x_i-m^u_i)^2}}{2(R_t^u)^2}
    \\
    &\leq \frac{\abs{c_i-x_i}(\abs{c_i-m^u_i}+\abs{x_i-m^u_i})}{2(R_t^u)^2},
\end{align*}
where $\vartheta_t^u = m^u_i+\sigma_t^u\sqrt{\theta_t^u}$.

Now, suppose $x$ and $c$ are on the opposite sides of the median $m^u$ in coordinate $i$, i.e. $(x_i-m^u_i)(c_i - m^u_i) < 0$. The threshold cut $(i,m^u_i+\sigma_t^u\sqrt{\theta_t^u})$ separates $x$ and $c$ if and only if $\sigma_t^u (x_i-m^u_i) \geq 0$, $\theta_t^u \leq (x_i-m^u_i)^2$ or $\sigma_t^u (c_i-m^u_i) \geq 0$, $\theta_t^u \leq (c_i-m^u_i)^2$. Thus, we have for every coordinate $i$ with $(x_i-m^u_i)(c_i-m^u_i) < 0$,
\begin{align*}
    \prob{\delta_x(i,\vartheta_t^u) \neq \delta_c(i,\vartheta_t^u) \mid T_t(u)} &=
    \frac{(c_i-m^u_i)^2+(x_i-m^u_i)^2}{2(R_t^u)^2}\\
    &\leq \frac{\abs{c_i-x_i}(\abs{c_i-m^u_i}+\abs{x_i-m^u_i})}{2(R_t^u)^2},
\end{align*}
where the last inequality follows from $\abs{c_i-x_i} \geq \max\{\abs{c_i-m^u_i},\abs{x_i-m^u_i}\}$, since $c_i, x_i$ are on the different sides of $m^u_i$.

Since the coordinate $i^u_t$ is chosen randomly and uniformly from $\{1,\cdots d\}$, the probability that $x$ and $c$ are separated at iteration $t$ is
\begin{align*}
    \pr[\delta_x(i_t^u,\vartheta_t^u) \neq \delta_c(i_t^u,\vartheta_t^u) \mid T_t(u)] 
    &\leq \sum_{i=1}^d \frac{\abs{c_i-x_i}(\abs{c_i-m^u_i}+\abs{x_i-m^u_i})}{2d\cdot (R_t^u)^2} \\
    &\leq \frac{\norm{c-x}_2(\norm{x-m^u}_2+\norm{c-m^u}_2)}{d\cdot (R_t^u)^2},
\end{align*}
where the last inequality follows from the Cauchy-Schwarz inequality and $(\abs{c_i}+\abs{x_i})^2 \leq 2c_i^2+2x_i^2$.

Then, the expected penalty is 
\begin{align*}
    \E[\phi_t^u(x)]
    &\leq \E\bigg[\prob{\delta_x(i_t^u,\vartheta_t^u) \neq \delta_c(i_t^u,\vartheta_t^u) \mid T_t(u)}\cdot 2R_t^u\bigg]\\
    &\leq \E\bigg[2\norm{c-x}_2\cdot\frac{\norm{c-m^u}_2+\norm{x-m^u}_2}{d\cdot R_t^u}\bigg].
\end{align*}
\end{proof}

To bound the expected penalty for point $x$, we consider two types of cuts based on three parameters: the maximum radius $R_t^u$ and distances $\norm{x - m^u}_2$, $\norm{c - m^u}_2$ between $x, c$ and the median $m^u$ . 
If $x$ is separated from its original center $c$ at iteration $t$ with
$$
R_t^u \leq \sqrt{\log_2 k} \cdot \max\{\norm{x-m^u}_2,\norm{c-m^u}_2\},
$$
then we call this cut a light cut. Otherwise, we call it a heavy cut.

\begin{lemma}\label{lem:penalty-L_2}
In every call of \textsc{Partition\_Leaf($u$)} (see Algorithm~\ref{alg:threshold_tree_kmedians-L_2}), the expected penalty for a point $x\in X$ is upper bounded as follows:
$$
\E\bigg[\sum_t \phi_t^u(x)\bigg] \leq O(\sqrt{\log k})\cdot \mathrm{cost}_{\ell_2}(x,C).
$$
\end{lemma}

\begin{proof}
If point $x$ is not separated from its original center $c$ in~\textsc{Partition\_Leaf($u$)}, then the total penalty is $0$. If $x$ is separated from its center $c$ in this call, then there are two cases: (1) the point $x$ is separated by a light cut; (2) the point $x$ is separated by a heavy cut. We first show that the expected penalty due to a heavy cut is at most $O(\sqrt{\log k})\cost_{\ell_2}(x,C)$. 

Denote the set of all heavy cuts at iteration $t$ in~\textsc{Partition\_Leaf($u$)} by $H_t^u$:
$$
H_t^u = \{x: \max\{\norm{x-m^u}_2,\norm{c-m^u}_2\} < R_t^u/ \sqrt{\log_2 k} \}.
$$
Then, by Lemma~\ref{lem:penalty}, the expected penalty $x$ incurs due to  a heavy cut is at most
$$
\E\Bigg[\sum_{t: x \in H_t^u} \phi_t^u(x)\Bigg] \leq
2\norm{x-c}_2\cdot \E\Bigg[\sum_{t:x\in H_t^u} \frac{\norm{x -m^u}_2+\norm{c-m^u}_2}{d\cdot R_t^u}\Bigg].
$$

Since the maximum radius $R_t^u$ is a non-increasing function of $t$, we split all steps of this call of \textsc{Partition\_Leaf} into phases with exponentially decreasing values of $R_t^u$. At phase $s$, the maximum radius $R_t^u$ is in the range $(R_1^u/2^{s+1},R_1^u/2^s]$, where $R_1^u$ is the maximum radius at the beginning of 
\textsc{Partition\_Leaf($u$)}.

Consider an arbitrary phase $s$ and step $t$ in that phase. Let $R = R_1^u/2^s$. For every center $c'$ with $\norm{c'-m^u}_2 \in (R/2,R]$, the probability that this center $c'$ is separated from the main part at step $t$ in phase $s$ is at least
$$
\prob{\delta_{c'}(i_t^u,\vartheta_t^u) \neq \delta_{m^u}(i_t^u,\vartheta_t^u) \mid T_t(u)}  
= \sum_{j=1}^d \frac{1}{d}\cdot \frac{(c'_j-m^u_j)^2}{2(R_t^u)^2} 
= \frac{\norm{c'-m^u}_2^2}{2d \cdot (R_t^u)^2} \geq \frac{1}{4d},
$$
where the last inequality is due to $\norm{c'-m^u}_2 > R/2 \geq R^u_t/2$ for step $t$ in the phase $s$.
Since there are at most $k$ centers, all centers with norm in $(R/2,R]$ are separated from the main part in at most $4d\ln k$ steps in expectation. Thus, the expected length of each phase is $O(d\log k)$ steps, and hence, the expected penalty $x$ incurred during phase $s$ is at most 
\begin{align*}
2\|x-c\|_2\cdot \E\bigg[\sum_{\substack{t: x\in H_t^u\\R^{u}_t \in (R/2,R]}} \frac{\norm{x-m^u}_2+\norm{c-m^u}_2}{d\cdot R_t^u}\bigg] 
 &\leq
2\|x-c\|_2\cdot \E\bigg[\sum_{\substack{t: x\in H_t^u\\R^{u}_t \in (R/2,R]}} \frac{\norm{x-m^u}_2+\norm{c-m^u}_2}{d\cdot R/2}\bigg] 
\\ &\leq
O(\log k) \cdot \|x-c\|_2\cdot \frac{\norm{x-m^u}_2+\norm{c-m^u}_2}{R}.
\end{align*}
Let $s'$ be the last phase for which
\begin{equation}\label{def:s-prime}
R^u_1/2^{s'}\geq \sqrt{\log_2 k}\cdot \max\{\|x-m^u\|_2,\|c-m^u\|_2\}.
\end{equation}
Then, in every phase $s>s'$, all cuts separating $x$ from its original center $c$ are light. Hence, the total expected penalty due to a heavy cut is upper bounded by 
\begin{multline*}
O(\log k) \cdot \|x-c\|_2\cdot (\norm{x-m^u}_2+\norm{c-m^u}_2) \cdot \sum_{s=0}^{s'}\frac{2^s}{R_1^u}
=\\=
O(\log k) \cdot \|x-c\|_2\cdot (\norm{x-m^u}_2+\norm{c-m^u}_2) \cdot \frac{2^{s'+1}}{R_1^u}.
\end{multline*}
Using the definition~(\ref{def:s-prime}) of $s'$, we write
$$
    (\norm{x-m^u}_2+\norm{c-m^u}_2) \cdot \frac{2^{s'+1}}{R_1^u}  
    \leq 2\frac{\|x-m^u\|_2+\|c-m^u\|_2}{R^u_1/2^{s'}}\leq \frac{4}{\sqrt{\log_2 k}}.
$$

Thus, the expected penalty due to a heavy cut is at most $O(\sqrt{\log k})\cost_{\ell_2}(x,C)$. 

We now analyze the expected penalty due to a light cut. Consider an iteration $t$ in~\textsc{Partition\_Leaf($u$)} with $x \not \in H_t^u$. By the analysis in Lemma~\ref{lem:penalty}, the probability that $x$ and $c$ are separated at iteration $t$ is at most
$$
\frac{\norm{c-x}_2(\norm{x-m^u}_2+\norm{c-m^u}_2)}{d\cdot (R_t^u)^2}.
$$
The probability that $x$ or $c$ is separated from the main part at iteration $t$ is at least
$$
\frac{\max\{\norm{x-m^u}_2^2,\norm{c-m^u}_2^2\}}{d(R_t^u)^2}.
$$
If $x$ or $c$ is separated from the main part, then the point $x$ will not incur penalty at any step after $t$. Thus, the probability that $x$ and $c$ are separated by a light cut in the end of~\textsc{Partition\_Leaf($u$)} is at most
$$
   \frac{\norm{c-x}_2(\norm{x-m^u}_2+\norm{c-m^u}_2)}{\max\{\norm{x-m^u}_2^2,\norm{c-m^u}_2^2\}} 
\leq 
   \frac{2\norm{c-x}_2}{\max\{\norm{x-m^u}_2,\norm{c-m^u}_2\}}. 
$$

Since the penalty of a light cut is at most $R_t^u \leq \sqrt{\log_2 k}\cdot \max\{\norm{x-m^u}_2,\norm{c-m^u}_2\}$, the expected penalty due to a light cut is at most $O(\sqrt{\log k})\cdot \mathrm{cost}_{\ell_2}(x,C)$.

This concludes the proof of Lemma~\ref{lem:penalty-L_2}.
\end{proof}

For every node $u$, the main part contains the median $m^u$, which is also the $\ell_1$-median of all centers in $X_u$. Thus, each cut sampled in the call \textsc{Partition\_Leaf($u$)} separates at most half of all centers in $X_u$ from the origin. The main part contains at most half of centers in $X_u$ at the end of the call \textsc{Partition\_Leaf($u$)}. Therefore, each leaf node generated in the end of \textsc{Partition\_Leaf($u$)} contains at most half of centers in $X_u$. Thus, the depth of the recursion tree is at most $O(\log k)$. By Lemma~\ref{lem:penalty-L_2} and Equation~(\ref{eq:l2-penalty}), we get the conclusion.
\end{proof}

%% file: algorithms/algorithm-L_2.tex
\begin{figure}[tb]
\begin{algorithmic}
\begin{framed}
\STATE {\bfseries Input:} a data set $X \subset \bbR^d$, centers 
$C=\{c_1,c_2,\dots, c_k\} \subset \bbR^d$
\STATE {\bfseries Output:} a threshold tree $T$
\STATE{}
\FUNCTION{\textsc{Main}$(X,C)$}
\STATE Create a root $r$ of the threshold tree $T$ containing $X_r = X\cup C$.
\STATE \textsc{Build\_tree}($r$).
\ENDFUNCTION 
\STATE{}
\FUNCTION{\textsc{Partition\_Leaf}($u$)}{
\STATE Compute the $\ell_1$ median $m^u$ of all centers in $X_u$.
\STATE Set the main part $u_0 = u$ and set $t=0$.
\WHILE{node $u_0$ contains more than $1/2$ of centers in $X_u$}
\STATE Update $t=t+1$.
\STATE Let $R_t^u = \max_{c \in X_{u_0}} \norm{c}_2$.
\STATE Sample $i_t^u \in \{1,2,\cdots,d\}$, $\theta_t^u \in [0,(R_t^u)^2]$, and $\sigma_t^u \in \{\pm 1\}$ uniformly at random.
\IF{two centers in $X_{u_0}$ are separated by $(i_t^u,m^u_i+\sigma_t^u\sqrt{\theta_t^u})$}{
\STATE Assign to $u_0$ two children 
$u_{\leq} = \{x\in X_{u_0}: x_i \leq \vartheta\}$ and 
$u_{>} = \{x\in X_{u_0}: x_i > \vartheta\}$ where $i = i_t^u,\vartheta = m^u_i+\sigma_t^u\theta_t^u$.
\STATE Update the main part $u_0$ be $u_{\leq}$ if $\sigma_t^u = 1$, and be $u_{>}$ otherwise (thus, the main part always contains $m^u$).
}\ENDIF
\ENDWHILE
}\ENDFUNCTION

\STATE{}
\FUNCTION{\textsc{Build\_tree}($u$)}{
\STATE Call \textsc{Partition\_Leaf}($u$).
\STATE Call \textsc{Build\_tree}($v$) for each leaf $v$ in the subtree of $u$ containing more than one center.
}\ENDFUNCTION 
\end{framed}
\end{algorithmic}

\caption{Threshold tree construction for $k$-medians in $L_2$}
\label{alg:threshold_tree_kmedians-L_2}
\end{figure}

%% file: paper/lower-bound.tex
\section{Lower Bound for Threshold Tree}\label{sec:lb-kmeans}

\subsection{Lower bound for \texorpdfstring{$k$}{k}-means}

In this section, we show a lower bound on the price of explainability for $k$-means.

\begin{theorem}\label{thm:lower_bound_kmeans}
For any $k$, there exists an instance $X$ with $k$ clusters such that the cost of explainable $k$-means clustering for every tree $T$ is at least
$$
\mathrm{cost}_{\ell_2^2}(X,T) \geq \Omega\rbr{\frac{k}{\log k}} \mathrm{OPT}_{\ell_2^2}(X).
$$
\end{theorem}

To prove this lower bound, we construct an instance as follows. We uniformly sample $k$ centers $C = \cbr{c^1,c^2,\cdots, c^k}$ from the $d$-dimensional unit cube $[0,1]^d$ where the dimension $d = 300 \ln k$. For each center $c^i$, we add two points $c^i \pm (\varepsilon,\varepsilon,\cdots,\varepsilon)$ with $\varepsilon = 300 \ln k /k$. We also add many points at each center such that the optimal centers for any threshold tree remain almost the same. Specially, we can add $k^2$ points co-located with each center $c^i$. Then, if one center $c^i$ is shifted by a distance of $\varepsilon$ in the threshold tree clustering, the cost of the co-located points at $c^i$ is at least $k^2\varepsilon^2$. Since the optimal regular cost for this instance is $kd\varepsilon^2$, the total cost of the threshold tree is lower bounded by $\Omega(k/\log k)\mathrm{OPT}_{\ell_2^2}(X)$. Consequently, we consider the threshold tree with optimal centers shifted by at most $\varepsilon$.  

First, we show that any two centers defined above are far apart with high probability.
\begin{lemma}\label{lem:separation}
With probability at least $1-1/k^2$ the following holds: The squared distance between every two distinct centers $c$ and $c'$ in $C$ is at least $d/12$. 
\end{lemma}
\begin{proof}
Consider any fixed two centers $c,c' \in C$. Since $c,c'$ are uniformly sampled from $[0,1]^d$, each coordinate of $c,c'$ is sampled from $[0,1]$; and centers $c,c'$ are sampled independently. Thus, we have
$$
\E_{c,c'} [\norm{c-c'}^2] = \sum_{i=1}^d \E_{c_i,c_i'} [(c_i-c_i')^2] = \frac{d}{6}.
$$
We use a random variable $X_i$ to denote $(c_i-c_i')^2$ for each coordinate $i \in \{1,\dots,d\}$. Since random variables $\{X_i\}_{i=1}^d$ are independent, by Hoeffding's inequality, we have

$$
\prob{\sum_{i=1}^d X_i - \E\bigg[\sum_{i=1}^d X_i\bigg]  \leq -\sqrt{2d\ln k}} \leq  e^{-4\ln k} = \frac{1}{k^4},
$$
where we used that $d = 300 \ln k$. This implies that the squared distance between $c$ and $c'$ is less than $d/12$ with probability at most $1/k^4$.
Using the union bound over all pairs of centers in $C$, we conclude that the squared distance between all pairs in $C$ is at least $d/12$ with probability at least $1-1/k^2$.
\end{proof}

If any two centers are far apart, then a point $x$ separated from its original center will incur a large penalty. Thus, we can get a lower bound if there exists an instance which satisfies: (1) any two centers are separated by a large distance; (2) every threshold tree separates a relatively large portion of points from their original centers. In particular, we prove that with probability $1-o(1)$, every threshold cut separates a relatively large portion of points from their original centers
in the random instance we constructed.

\begin{lemma}\label{lem:lb_cut}
With probability at least $1-1/k^2$, the following holds: every threshold cut $(i,\theta)$ with $i\in \{1,2,\cdots,d\}$ and $\theta \in [0,1)$ separates at least $\varepsilon k/4$ points from their original centers.
\end{lemma}

\begin{proof}
Consider a fixed coordinate $i \in \{1,\dots,d\}$. We project each center and its rectangular neighborhood onto this coordinate. For each center $c^j \in C$, we define an interval $I^j_i$ as the intersection of $[0,1]$ and the $\varepsilon$-neighborhood of its projection $c^j_i$, i.e. $I^j_i = (c^j_i - \varepsilon, c^j_i + \varepsilon) \cap [0,1]$. Each interval $I^j_i$ has length at least $\varepsilon$. If we pick a threshold cut inside any interval $I^j_i$, then we separate at least one points from center $c^j$. In this case, the interval $I^j_i$ is called covered by this threshold cut. Then, we give the lower bound on the minimum number of intervals covered by a threshold cut. 

For a fixed set of centers $C$, we consider at most $2k$ special positions for the threshold cut at coordinate $i$ as follows. Let $E_i$ be the set containing two end points of intervals $I^j_i$ for all centers $c^j$. For any threshold cut at coordinate $i$, the closest position in set $E_i$ covers exactly the same set of intervals as this threshold cut. Thus, we only need to consider threshold cuts at positions in $E_i$. 

For centers chosen uniformly from $[0,1]^d$, the set $E_i$ contains $2k$ random variables. Suppose we pick a threshold cut at a position $\theta$ in $E_i$ related to interval $I^j_i$. Conditioned on the position $\theta$, the other $k-1$ centers $c^j$ for $j\neq j^*$ are uniformly distributed in $[0,1]^d$ since all centers are chosen independently. For $j \in \{1,2,\cdots,k\}\setminus \{j^*\}$, let $Y^j_i$  be the indicator random variable that the interval $I^j_i$ contains this position $\theta$. For each variable $Y^j_i$, we have $\varepsilon \leq \prob{Y^j_i = 1} \leq 2\varepsilon$. Since random variables $Y^j_i$ are independent, by the Chernoff bound for Bernoulli random variables, we have
$$
\prob{\sum_{j} Y^j_i - \E\Bigg[\sum_{j} Y^j_i\Bigg] \leq -\sqrt{18\varepsilon k \ln k} \mid \theta } \leq e^{-4\ln k} = \frac{1}{k^4}.
$$
Thus, we have the number of intervals containing this position $\theta$ is at least $\varepsilon k/4$ with probability at least $1-1/k^4$.

Since we have $2k$ positions $E_i$ for each coordinate $i \in \{1,2,\cdots,d\}$, there are total $2dk$ positions for threshold cuts. Using the union bound over all positions, we have the minimum number of intervals covered by a threshold cut is at least $\varepsilon k/4$ with probability at least $1-1/k^2$. Since the threshold cut separates one point from its original center for each covered interval, we have every threshold cut separates at least $\varepsilon k/4$ points from their original centers in this case.
\end{proof}

\begin{proof}[Proof of Theorem~\ref{thm:lower_bound_kmeans}]

By Lemma~\ref{lem:separation}, we can only consider the instance where any two centers are separated with the squared distance at least $d/12$. Note that the optimal centers for any threshold tree remain almost the same as centers $C$. Thus, we analyze the $k$-means cost given by any threshold tree with respect to center $C$. If a point in $X$ is separated from its original center, this point will finally be assigned to another center in $C$. By the triangle inequality, the $k$-means cost of this point is at least $d/20$.  By Lemma~\ref{lem:lb_cut}, there exists an instance such that any threshold cut separates at least $\varepsilon k/4$ points from their original centers.
Thus, there exists an instance $X$ such that any threshold tree $T$ has the $k$-means cost at least
$$
\mathrm{cost}_{\ell_2^2}(X,T) \geq \frac{\varepsilon k}{4}\cdot \frac{d}{20} = \frac{\varepsilon kd}{80}.
$$
Note that the optimal regular $k$-means cost for this instance $X$ is  
$$
\mathrm{OPT}_{\ell_2^2}(X) =2k \cdot \varepsilon^2d.
$$
Therefore, the $k$-means cost for this instance $X$ given by any threshold tree $T$ is at least
$$
\mathrm{cost}_{\ell_2^2}(X,T) \geq \frac{1}{160\varepsilon}\cdot \mathrm{OPT}_{\ell_2^2}(X) 
= \Omega\rbr{\frac{k}{\log k}}\cdot \mathrm{OPT}_{\ell_2^2}(X) .
$$
\end{proof}


\subsection{\texorpdfstring{Lower bound for $k$-medians in $\ell_2$}{Lower bound for k-medians in l2}}

In this section, we show a lower bound on the price of explainability for $k$-medians in $\ell_2$.

\begin{theorem}\label{thm:lower_bound_kmedian}
For every $k\geq 1$, there exists an instance $X$ with $k$ clusters such that the $k$-medians with $\ell_2$ objective cost of every threshold tree $T$ is at least
$$
\mathrm{cost}_{\ell_2}(X,T) \geq \Omega(\log k) \mathrm{OPT}_{\ell_2}(X).
$$
\end{theorem}

To prove this lower bound, we use the construction similar to that used in Theorem~\ref{thm:lower_bound_kmeans}. We discretize the $d$-dimensional unit cube $[0,1]^d$ into grid with length $\varepsilon = 1/\ceil{\ln k}$, where the dimension $d = 300 \ln k$. We uniformly sample $k$ centers $C = \{c^1,c^2,\cdots,c^k\}$ from the above grid $\{0,\varepsilon,2\varepsilon,\cdots,1\}^d$. For each center $c^i$, we add $2$ points $c^i \pm (\varepsilon,\varepsilon,\cdots,\varepsilon)$ to this center. Similar to Theorem~\ref{thm:lower_bound_kmeans}, we also add many points at each center such that the optimal centers for any threshold tree remain almost the same. 

Similar to Lemma~\ref{lem:separation}, we show that any two centers defined above are far apart with high probability.
\begin{lemma}\label{lem:separation_l2}
With probability at least $1-1/k^2$ the following holds: The distance between every two distinct centers $c$ and $c'$ in $C$ is at least $\sqrt{d}/4$. 
\end{lemma}

\begin{proof}
To sample a center from the grid uniformly, we can first sample a candidate center uniformly from the cube $[-\varepsilon/2,1+\varepsilon/2]^d$ and then move it to the closest grid point. Note that the $\ell_2$-distance from every point in this cube to its closest grid point is at most $\varepsilon\sqrt{d} = o(1)$. By Lemma~\ref{lem:separation}, the $\ell_2$ distance between every pairs of candidate centers is at least $\sqrt{d/12}$ with probability at least $1-1/k^2$. Thus, the distance between every two distinct centers is at least $\sqrt{d}/4$ with probability at least $1-1/k^2$.
\end{proof}

For every node in the threshold tree, we can specify it by threshold cuts in the path from the root to this node.  Thus, we define a path $\pi$ as an ordered set of tuples $(i_j,\theta_j,\sigma_j)$, where $(i_j,\theta_j)$ denotes the $j$-th threshold cut in this path and $\sigma_j \in \{\pm 1\}$ denotes the direction with respect to this cut. We use $u(\pi)$ be the node specified by the path $\pi$. We define a center is damaged if one of its two points are separated by this cut, otherwise a  center is undamaged. Let $F_u$ be the set of undamaged centers in node $u$. 

\begin{lemma}\label{lem:damage}
With probability at least $1-1/k$, the following holds: For every path $\pi$ with length less than $\log_2 k/4$, we have (a) the node $u(\pi)$ contains at most $\sqrt{k}$ undamaged centers; or (b) every cut in node $u(\pi)$ damages at least $\varepsilon|F_{u(\pi)}|/2$ centers in $F_{u(\pi)}$.  
\end{lemma}

\begin{proof}
Consider any fixed path $\pi$ with length less than $\log_2 k/4$. We upper bound the probability that both events (a) and (b) do not happen conditioned on $F_{u(\pi)}$.  If $|F_{u(\pi)}| \leq \sqrt{k}$, then the event (a) happens. For the case $F_{u(\pi)}$ contains more than $\sqrt{k}$ centers, we pick an arbitrary threshold cut $(i,\theta)$ in the node $u(\pi)$. For every center $c$ in $F_{u(\pi)}$, the probability we damage this center $c$ is at least $\varepsilon$. Let $X_j$ be the indicator random variable that the $j$-th center in $F_{u(\pi)}$ is damaged by the threshold cut $(i,\theta)$. Then, we have the expected number of centers in $F_{u(\pi)}$ damaged by this cut $(i,\theta)$ is 
$$
\E\bigg[\sum_{j} X_j\bigg] \geq \varepsilon \abs{F_{u(\pi)}}.
$$
Let $\mu = \E[\sum_{j} X_j]$. By the Chernoff bound for Bernoulli random variables, we have
$$
    \prob{\sum_j X_j \leq \varepsilon\abs{F_{u(\pi)}}/2} \leq \prob{\sum_j X_j \leq \mu/2} 
    \leq e^{-\mu/8} \leq e^{-\varepsilon\sqrt{k}/8}.
$$
Using the union bound over all threshold cuts in $u(\pi)$, the failure probability that both event (a) and (b) do not happen is at most $e^{-\varepsilon\sqrt{k}/16}$. The number of paths with length less than $\log_2 k/4$ is at most $m(2d/\varepsilon)^m \leq e^{-\log^2 k}$. Thus, by the union bound over all paths with length less than $\log_2 k /4$, we get the conclusion.  
\end{proof}

\begin{proof}[Proof of Theorem~\ref{thm:lower_bound_kmedian}]
By Lemma~\ref{lem:separation_l2} and Lemma~\ref{lem:damage}, we can find an instance $X$ such that both two properties hold. We first show that the threshold tree must separate all centers. Suppose there is a leaf contains more than one center. Since the distance between every two centers is at least $\sqrt{d}/4$ and there are many points at each center, the cost for this leaf can be arbitrary large. To separate all centers, the depth of the threshold tree is at least $\ceil{\log_2 k}$.

We now lower bound the cost for every threshold tree that separates all centers. Consider any threshold tree $T$ that separates all centers. We consider the following two cases. If the number of damaged centers at level $\floor{\log_2 k}/4$ of threshold tree $T$ is more than $k/2$, then the cost given by $T$ is at least
$$
\mathrm{cost}_{\ell_2}(X,T) \geq \frac{k}{2}\cdot\frac{\sqrt{d}}{8} = \frac{k\sqrt{d}}{16}.
$$
If the number of damaged centers at level $\floor{\log_2 k}/4$ of threshold tree $T$ is less than $k/2$, then the number of undamaged centers at every level $i= 1,2,\dots,\floor{\log_2 k}/4$ is at least $k/2$. We call a node $u$ a small node if it contains at most $\sqrt{k}$ undamaged centers, otherwise we call it a large node. Then, we lower bound the number of damaged centers generated at any fixed level $i \in \{1,2,\cdots,\floor{\log_2 k}/4\}$. Since the number of nodes at level $i$ is at most $k^{1/4}$, the number of undamaged centers in small nodes at level $i$ is at most $k^{3/4}$. Thus, the number of undamaged centers in large nodes at level $i$ is at least $k/4$. By Lemma~\ref{lem:damage}, the number of damaged centers generated at level $i$ is at least $\varepsilon k/8$. Therefore, the cost given by this threshold tree $T$ is at least
$$
\mathrm{cost}_{\ell_2}(X,T) \geq \frac{\floor{\log_2 k}}{4} \frac{\varepsilon k}{8} \frac{\sqrt{d}}{8} = \Omega(k\sqrt{d}\varepsilon \log k).
$$
Note that the optimal cost for this instance is at most $k\varepsilon\sqrt{d}$ and $\varepsilon = 1/\ceil{\log k}$. Combining the two cases above, we have the cost given by threshold tree $T$ is at least
$$
\mathrm{cost}_{\ell_2}(X,T) = \Omega(k\sqrt{d}\varepsilon \log k) = \Omega(\log k) \mathrm{OPT}_{\ell_2}(X).
$$
\end{proof}

%% file: paper/fast-algorithm.tex
\section{Fast Algorithm}\label{sec:fast_alg}
In this section, we provide a fast variant of Algorithm~\ref{alg:threshold_tree_kmedians} with running time $O(kd\log^2 k)$. The input of this algorithm is the set of reference centers $c^1,\dots,c^k$ and the output is a threshold tree that splits all centers. The algorithm does not consider the data points (hence, it does not explicitly assign them to clusters). It takes an extra $O(nk)$ time to assign every point in the data set to one of the leaves of the threshold tree.

This fast variant of Algorithm~\ref{alg:threshold_tree_kmedians} picks a separate threshold cut $\omega^u$ for each leaf $u$. This cut is chosen uniformly at random from $R^u$, where 
$$R^u = \bigcup_{c^i,c^j \in X_u} S_{ij}.$$ 
That is, $R^u$ is the set of all cuts $\omega$ that separate at least two centers in $X_u$. The algorithm then splits leaf $u$ into two parts using $\omega^u$.

A straightforward implementation of the algorithm partitions each leaf by computing $\delta_c(\omega^u)$ for all centers $c$ in $X_u$. It takes $O(d\cdot |X_u \cap C|)$ time to find $R^u$ and sample $\omega^u$ for each $u$.
It takes time $O(|X_u \cap C|)$  to split 
$X_u$ into two groups. Thus, the total running time of this implementation of the algorithm is $O(k^2 d)$. We now discuss how to implement this algorithm with running time $O(kd\log^2 k)$ using red-black trees. 

The improved algorithm stores centers for each leaf of the threshold tree in $d$ red-black trees. Centers in the $i$-th red-black tree are sorted by the $i$-th coordinate. Using red-black trees, we can find the minimum and maximum values of $c_i$ for $c\in C\cap X_u$ in time $O(d\log k)$. Denote these values by $a_i$ and $b_i$, then
$$R^u = \bigcup_i\; \{i\}\times [a_i, b_i].$$
Hence, we can find $R^u$ and sample a random cut $\omega^u$ in time $O(d\log k)$ for each $u$.

To partition set $X^u$ into two groups with respect to $\omega^u=(i,\theta)$, we consider the $i$-th red-black tree for leaf $u$ and find the sizes of the new parts, $Left = \{c\in X_u\cap C: c_i \leq \theta\}$ and $Right = \{c\in X_u\cap C: c_i > \theta\}$. We choose the set that contains fewer centers. Let us assume that the second set ($Right$) is smaller the first one ($Left$). Then, we find all centers in $Right$ and delete them from this red-black tree and all other red-black trees for node $u$. We assign the updated red-black trees (with deleted $Right$) to the left child of $u$. For the right child, we build $d$ new red-black trees, which store centers for $Right$. Since we delete at most half of all centers in the red-black tree, each center is deleted at most $O(\log k)$ times. Each time it is deleted from $d$ trees and inserted into $d$ trees. Each deletion and insertion operation takes time $O(\log k)$. Thus, the total time of all deletion and insertion operations is $O(kd\log^2 k)$. 

We note that though this algorithm slightly differs from the algorithm presented in Section~\ref{sec:k-medians}, its approximation guarantees are the same.